\documentclass[sigconf]{acmart}

\setcopyright{none}
\renewcommand\footnotetextcopyrightpermission[1]{}

\usepackage{algorithm}
\usepackage[noend]{algpseudocode}
\usepackage[capitalize,nameinlink]{cleveref}
\graphicspath{{figures/}}

\newcommand{\R}{\mathbb{R}}
\newcommand{\norm}[1]{\lVert #1 \rVert}
\newcommand{\abs}[1]{\lvert #1 \rvert}

\newcommand{\enc}[1]{[\![#1]\!]}
\newcommand{\Enc}{\mathsf{Enc}}
\newcommand{\Dec}{\mathsf{Dec}}
\newcommand{\rot}[1]{\mathrm{rot}_{#1}}
\newcommand{\ctadd}{\oplus}
\newcommand{\ctmul}{\otimes}
\newcommand{\ptmul}{\odot}

\newcommand{\xhat}{\hat{x}}
\newcommand{\region}{\mathcal{D}}
\newcommand{\errbox}{\mathcal{E}}
\newcommand{\graph}{\mathcal{G}}
\newcommand{\offsets}{\mathcal{S}}
\newcommand{\tol}{\rho_{\mathrm{tol}}}
\newcommand{\tcyc}{t_{\mathrm{cyc}}}
\newcommand{\tref}{t_{\mathrm{ref}}}
\newcommand{\Lref}{L_{\mathrm{ref}}}
\newcommand{\dret}{d_{\mathrm{ret}}}
\newcommand{\dcmd}{d_{\mathrm{cmd}}}
\newcommand{\dpred}{d_{\mathrm{pred}}}
\newcommand{\dphi}{d_{\varphi}}
\newcommand{\nb}{n_{\mathrm{b}}}
\newcommand{\etaref}{\eta_{\mathrm{ref}}}
\newcommand{\etacyc}{\eta_{\mathrm{cyc}}}
\newcommand{\etakey}{\eta_{\mathrm{key}}}
\newcommand{\sigfl}{\sigma_{\mathrm{fl}}}
\newcommand{\rstar}{r^{\star}}

\AtBeginDocument{\theoremstyle{acmdefinition}\newtheorem{assumption}[theorem]{Assumption}\newtheorem{remark}[theorem]{Remark}}
\crefname{assumption}{Assumption}{Assumptions}
\Crefname{assumption}{Assumption}{Assumptions}
\crefname{remark}{Remark}{Remarks}
\Crefname{remark}{Remark}{Remarks}
\newcommand{\abar}{\bar{\alpha}}
\newcommand{\rg}{r_{\mathrm{g}}}
\newcommand{\rpos}{r^{\star}_{\mathrm{pos}}}
\newcommand{\dhat}{\hat{d}}
\newcommand{\ghat}{\hat{g}}
\newcommand{\etag}{\eta_{\mathrm{g}}}
\newcommand{\etapos}{\eta_{\mathrm{pos}}}

\begin{document}

\title{Control and Cryptosystem Co-Design for Encrypted Nonlinear Multi-Agent Coordination}

\author{Sai Sandeep Damera}
\email{sdamera@umd.edu}
\affiliation{%
  \institution{University of Maryland}
  \city{College Park}
  \state{Maryland}
  \country{USA}}

\author{Maria Charitidou}
\email{mchar@umd.edu}
\affiliation{%
  \institution{University of Maryland}
  \city{College Park}
  \state{Maryland}
  \country{USA}}

\author{Asim Zoulkarni}
\email{asimz@umd.edu}
\affiliation{%
  \institution{University of Maryland}
  \city{College Park}
  \state{Maryland}
  \country{USA}}

\author{John S. Baras}
\email{baras@umd.edu}
\affiliation{%
  \institution{University of Maryland}
  \city{College Park}
  \state{Maryland}
  \country{USA}}

\begin{abstract}
A cloud that must not see the agents' states can steer them on data encrypted under CKKS, the approximate scheme for real-valued arithmetic. Its stored estimate admits only so many multiplications before a refresh that costs time and error, and every decrypted command is flooded with noise of a mandated size, constraints no controller designed in exact arithmetic meets. For networked agents with polynomial dynamics, we prove one closed-loop bound tying the controller and the cryptosystem in both directions. If the exact loop contracts on a region, the encrypted loop stays within a certified radius of it, set by the per-cycle and refresh errors, the refresh interval, the flooding and the delay. The errors and the bound's probability are derived from the library's circuits, not measured. Inverted, the bound turns the tolerance and a failure-probability target into the refresh interval, and with it the precision, security parameters and period. Conversely, costs shape the realization: a tracing backend tallies any pipeline's operations before a ciphertext exists, and evaluating the prediction about each agent's reference holds the released error below its sixteen-agent level. A full satellite-train mission encrypted at 128-bit security ran inside its certified radius, and the certificate holds across six graphs and at 32, 64 and 128 agents. Security enters the design twice, in time and precision. The cryptosystem is therefore part of the control design, derived from the same specification and certified by the same guarantee as the controller.
\end{abstract}


\maketitle
\pagestyle{plain}

\section{Introduction}
\label{codesign:sec:introduction}

A fleet of agents is coordinated by a cloud that must not learn where any of them is. Each agent sends the cloud only encrypted measurements. The cloud estimates the fleet's state, predicts it forward and computes every agent's command without decrypting anything, and each agent decrypts only the command addressed to it. Homomorphic encryption makes this possible. The approximate scheme for real-valued arithmetic, CKKS~\cite{cheon2017homomorphic}, is the one encrypted control uses for it~\cite{schulzedarup2021encrypted,schluter2024code}. The scheme admits three operations on encrypted vectors, addition, multiplication and cyclic rotation, and every multiplication spends one of a finite number of levels. A value that stays encrypted for longer than its levels last has to be refreshed by bootstrapping, an operation that takes tens of seconds at standard security and perturbs the value it refreshes.

\begin{figure*}[t]
\centering
\includegraphics[width=\textwidth]{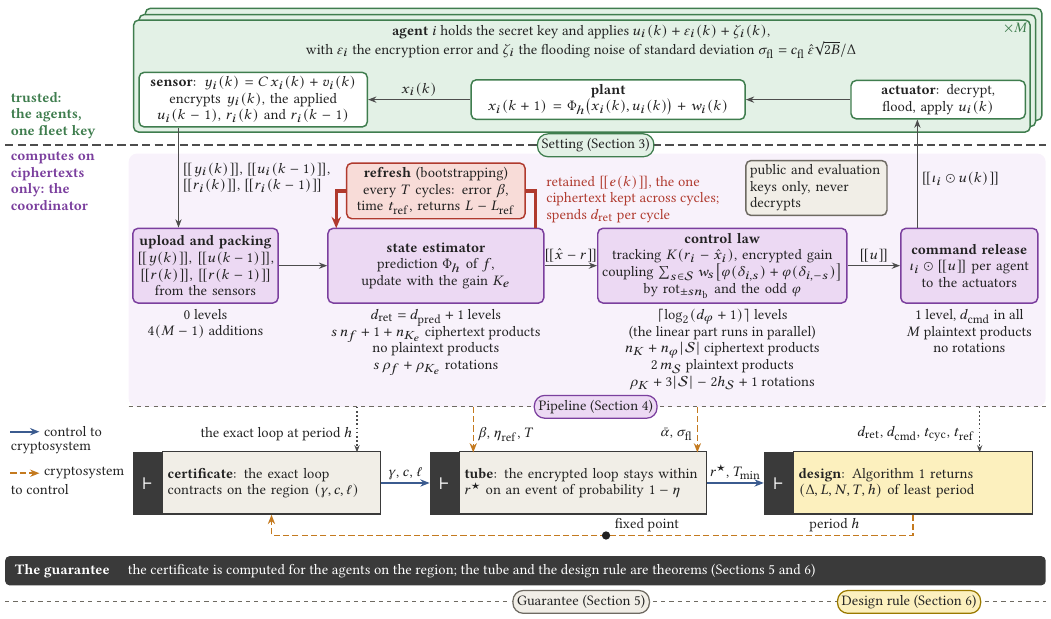}
\caption{The agents (top), the coordinator's cycle with each stage's levels and counts (middle), and the guarantee and the design rule (bottom).}
\label{codesign:fig:overview}
\end{figure*}

The coordinator we study keeps such a value. Its estimate of the fleet is updated recursively on ciphertexts for the whole mission, so it is refreshed many times, and each refresh adds an error with a failure probability of its own. This separates a coordinator that estimates from one that only evaluates a coupling on fresh uploads, and it ties the two systems together. The tolerance of the mission sets how often the estimate may be refreshed, because every refresh adds error, and the arithmetic depth of one control cycle sets how often it must be, because levels run out. The depth sets the ring dimension at a given security level, and the ring sets the time of a refresh and of every operation. That time sets the control period, and the period sets how much the closed loop contracts per cycle, which is where the tolerance came in. Neither system can be set after the other. The designs that tie them do so through a single quantity (\cref{codesign:sec:related}). It is a key length chosen with a static gain~\cite{teranishi2022designing}, a reset horizon chosen with a linear controller and its word size~\cite{murguia2020secure}, or a bootstrapping period fixed in advance with the closed loop analysed under it~\cite{schlor2024bootstrapping}. Their controllers are linear, and their encrypted state, where one is kept, is reset or bounded by the controller's design rather than refreshed. For networked nonlinear agents with a retained encrypted estimate the constraints run in both directions.

We state the two directions as one closed-loop bound, a tube, and turn it into a design procedure. The encrypted loop stays within a radius of the loop in exact arithmetic, and the radius is a function of the per-cycle encryption error, the refresh error and the refresh interval through the contraction rate of the exact loop. Each constant in it is derived from the agents, the graph and the library's own circuits, not measured, so the bound can be inverted. Inverting it gives the fewest cycles between refreshes from the tolerance. The design is the fixed point of the cycle above, found by a finite enumeration, and we present it as an algorithm whose output comes with the guarantee (\cref{codesign:sec:design}). Security enters a second time as a precision constraint. Decryptions of the approximate scheme leak the secret key unless each released value is flooded with noise of a mandated size~\cite{li2021security,li2022securing}, a multiple of the released command's encryption error, so the derived errors size it and the loop's remaining margin must absorb it. The realization runs the other way. The cost model of the scheme decides how each stage of the loop is computed, which values are retained, and what the agents upload. Because the coordinator's computation does not depend on its data, a tracing backend records the exact sequence of operations of any pipeline before a single ciphertext exists.

\Cref{codesign:fig:overview} draws the framework. The contributions are the following.
\begin{enumerate}
\item An encrypted pipeline for networked agents with polynomial dynamics, in which the coordinator retains and refreshes the fleet estimate (\cref{codesign:sec:pipeline,codesign:app:counts}). Every stage's depth, products and rotations are stated in closed form and reproduced by a tracing backend. The prediction is evaluated about each agent's reference, which holds the released error below its sixteen-agent level at 32, 64 and 128 agents.
\item A closed-loop guarantee (\cref{codesign:sec:guarantee}). Under a contraction certificate for the exact loop, the encrypted loop stays within a derived radius of it on an event of derived probability. The flooding enters through the loop's response to it, and the computation delay is a state.
\item The design rule as an algorithm (\cref{codesign:sec:design}). It runs from the tolerance, the mission, the failure-probability target, the agents, the graph and the security level to the scaling factor, the levels, the ring dimension, the refresh interval and the period. A theorem states what its output satisfies.
\end{enumerate}

\Cref{codesign:sec:evaluation} backs the three with the full mission run encrypted at the library's 128-bit setting on the ring of sixteen against the certified radius. The ring at 32, 64 and 128 agents runs encrypted for twelve cycles at that setting against the certified radius of each size. Every coordination graph is traced, timed at that setting and designed by the rule.

\section{Related Work}
\label{codesign:sec:related}

\paragraph{Encrypted control.} Encrypted linear controllers were first built on partially homomorphic schemes, multiplicative~\cite{kogiso2015cyber} and additive~\cite{farokhi2016secure}, and extended to cloud-based model predictive control and estimation~\cite{alexandru2018cloud,alexandru2019encrypted,alexandru2020secure}. The tutorials of Schulze Darup et al.~\cite{schulzedarup2021encrypted} and Schl\"uter et al.~\cite{schluter2024code} fix the vocabulary we use: the cloud and the client, plaintext and ciphertext, the level as what one multiplication consumes, and the encryption error as a bounded disturbance that closed-loop stability suppresses. Dynamic controllers over lattice schemes keep an encrypted state across cycles and bound its growth by controller design~\cite{kim2023dynamic}, and a cloud has run multi-agent optimization~\cite{binfet2023towards} and system identification~\cite{adamek2024encrypted} on CKKS-encrypted data. In all of these the cryptosystem's parameters are chosen after the controller is designed, from a bound it provides.

\paragraph{Bootstrapping and parameter selection in control.} Kim et al.~\cite{kim2016encrypting} compared the time of a refresh with the sampling period of one linear controller. Schlor and Allg\"ower~\cite{schlor2024bootstrapping} kept a linear controller's state encrypted and modelled each bootstrapping step as a sector-bounded error on it. Their matrix-inequality test gives stability and the $\ell_2$ gain under a bootstrapping period fixed in advance, with the scale, the levels and the ring dimension not designed and no sampling period in the model. Parameters have been chosen from a tolerance for linear controllers in integer schemes without bootstrapping~\cite{kim2023dynamic}. Murguia et al.~\cite{murguia2020secure} designed the reset horizon, the controller and the quantizer together for an additively homomorphic scheme in exact arithmetic. The word size follows from them, and the encrypted state is reset to zero at each horizon. Teranishi et al.~\cite{teranishi2022designing} chose the key length of a multiplicatively homomorphic scheme together with a static gain from an identification-time criterion, a more strongly stabilizing gain admitting a shorter key, with no controller state and no closed-loop bound on the encryption error. The pipeline of~\cite{damera2026end} ran the full loop of linear agents on CKKS with a steady-state gain and a bootstrapping period set by the depth alone. It applied the graph Laplacian by the diagonal method~\cite{halevi2014algorithms}, at a cost set by the number of nonzero cyclic diagonals rather than by circulancy. None of these sets the scale, the levels, the ring dimension, the refresh interval and the period of an approximate scheme from the tolerance of a nonlinear networked loop whose encrypted estimate is retained and refreshed. None ties the period to the measured cost of the cycle. \Cref{codesign:sec:design} does both.

\paragraph{Numerics and error analysis of the scheme.} Schlottke-Lakemper and Kholod~\cite{schlottkelakemper2024securecompute} showed that the three operations of the scheme suffice for finite-difference simulation on encrypted grids. They introduced the cyclic-shift construction, two rotations and a pair of plaintext masks, for neighbour access in an encrypted vector that does not fill the ciphertext. The refresh of CKKS ciphertexts was introduced by Cheon et al.~\cite{cheon2018bootstrapping} and made faster and more precise by Bossuat et al.~\cite{bossuat2021bootstrapping} and Lee et al.~\cite{lee2022highprecision}. The error of the scheme's operations has been analysed with worst-case and average-case bounds~\cite{costache2020heuristic,costache2023precision,murphy2024centrallimit}. Our refresh error is derived for the library we run, OpenFHE~\cite{albadawi2022openfhe}, from its moduli, correction factor and modular-reduction polynomial (\cref{codesign:app:refresh}). Li and Micciancio~\cite{li2021security} showed that decryptions of the approximate scheme leak the secret key, a leak closed by flooding each decryption with noise sized by the released value's error~\cite{li2022securing}. Guo et al.~\cite{guo2024keyrecovery} showed that flooding sized below the worst case falls to key recovery, and Cheon et al.~\cite{cheon2024attacks} carried the attacks to exact schemes. We take the mandated size as a constraint on the control precision.

\section{Setting}
\label{codesign:sec:setting}

\subsection{Agents, graph and specification}
\label{codesign:subsec:agents}

Consider $M$ identical agents with state $x_i \in \R^n$ and command $u_i \in \R^m$, $i = 1, \dots, M$, whose dynamics $\dot{x}_i = f(x_i) + B u_i$ have a polynomial vector field $f$ of degree $d_f$. The class is polynomial because the coordinator evaluates $f$ on ciphertexts, and the scheme computes additions and multiplications and nothing else. A smooth field is replaced by its polynomial approximation on the region of operation, and the approximation error is part of the disturbance below. The degree sets the depth of one stage of the prediction (\cref{codesign:subsec:estimator}), which is where the class meets the cryptosystem. On a compact convex region $\region \subset \R^n$ the field is Lipschitz with a constant $\ell_f$, and every result below is stated on $\region$.

The worked example is a train of satellites on a low circular orbit, coordinated to hold their along-track spacing. The state of agent $i$ is its position $q_i = (q_{i,1}, q_{i,2}, q_{i,3})$ relative to a reference point on the orbit, radial, along-track and cross-track, together with its velocity $\dot{q}_i$, so $n = 6$, and its command is the acceleration of its thrusters, $m = 3$. With $\omega$ the mean motion of the reference orbit and $R_0$ its radius, the relative motion under two-body gravity in the local orbital frame~\cite{alfriend2010-nonlinear-spacecraft}, expanded to second order in $q_i / R_0$, is
\begin{align}
\ddot{q}_{i,1} &= 2\omega \dot{q}_{i,2} + 3\omega^2 q_{i,1} + \frac{\omega^2}{R_0}\Bigl(-3 q_{i,1}^2 + \tfrac{3}{2} q_{i,2}^2 + \tfrac{3}{2} q_{i,3}^2\Bigr) + u_{i,1}, \nonumber \\
\ddot{q}_{i,2} &= -2\omega \dot{q}_{i,1} + \frac{3\omega^2}{R_0}\, q_{i,1} q_{i,2} + u_{i,2}, \label{codesign:eq:satellite} \\
\ddot{q}_{i,3} &= -\omega^2 q_{i,3} + \frac{3\omega^2}{R_0}\, q_{i,1} q_{i,3} + u_{i,3}, \nonumber
\end{align}
a field of degree $d_f = 2$ whose linear part is the Clohessy and Wiltshire model. The reference orbit has $R_0 = 6\,878$ km and $\omega = 1.107 \times 10^{-3}$ rad/s, a period of 95 minutes. The region is $\abs{q_{i,1}}, \abs{q_{i,3}} \le 10$ km, $\abs{q_{i,2}} \le 200$ km and $\norm{\dot{q}_i} \le 20$ m/s. At the region's corner the quadratic terms are 28 percent of the gravity-gradient acceleration, so the instance is nonlinear where the train is long, and the third-order remainder, which is part of the disturbance, stays below 0.7 percent of it. Each agent measures its own relative position, $p = 3$ and $C = [\,I_3 \;\; 0\,]$. The mission holds the train at a spacing of 20 km for one orbit, contracts it to 10 km over the next three orbits along a quintic schedule, and holds it for two more, six orbits in all (\cref{codesign:fig:loops}). The reference $r_i$ of agent $i$ is its place in that train, $i - (M + 1)/2$ spacings from the train's centre, with the radial offset that sustains the schedule's along-track rate under two-body gravity. The agent applies the schedule's feedforward acceleration itself, so that the coordinator's command is the correction. We run fleets of $M = 16$, 32, 64 and 128.

The coordinator acts with a period $h$ and holds each command over it. The sampled dynamics are
\begin{equation}
x_i(k+1) = \Phi_h\bigl(x_i(k), u_i(k)\bigr) + w_i(k), \qquad y_i(k) = C x_i(k) + v_i(k),
\label{codesign:eq:sampled}
\end{equation}
where $\Phi_h$ is the explicit integrator with $s$ stages that the coordinator also uses to predict (\cref{codesign:subsec:estimator}), so that $w_i$ collects the integrator's defect over one period together with the process disturbance, and $v_i$ is the measurement noise. Both are bounded, $\norm{w_i(k)} \le \bar{w}$ and $\norm{v_i(k)} \le \bar{v}$. In the worked example the measurement noise is Gaussian with a standard deviation of 3 m per axis, clipped at three standard deviations, and the process disturbance is an acceleration of at most $10^{-6}$ m/s$^2$ per axis.

The agents coordinate over an undirected graph $\graph$ on the $M$ agents with positive edge weights. In the agents' order, the weighted adjacency matrix of $\graph$ has nonzero cyclic diagonals at a set of offsets $\offsets \subset \{1, \dots, \lfloor M/2 \rfloor\}$, so that agent $i$ is coupled to agent $i \pm s$, indices modulo $M$, only if $s \in \offsets$. The ring is $\offsets = \{1\}$ and the $k$-nearest ring $\offsets = \{1, \dots, k\}$, both circulant, every cyclic diagonal being constant. The $\sqrt{M} \times \sqrt{M}$ torus with periodic boundaries, in row-major order, is $\offsets = \{1, \sqrt{M} - 1, \sqrt{M}\}$, the hypercube $\offsets = \{1, 2, 4, \dots, M/2\}$ and the complete graph $\offsets = \{1, \dots, \lfloor M/2 \rfloor\}$. The torus and the hypercube are not circulant, some of their cyclic diagonals containing zeros (\cref{codesign:rem:general_graphs}). The number of nonzero cyclic diagonals, not the number of agents, sets the cost of the coupling (\cref{codesign:subsec:law}). The mission gives every agent a reference $r_i(k)$. The coordination law drives the formation errors $e_i = x_i - r_i$ to zero through a local term and a coupling on the differences of the formation errors,
\begin{equation}
u_i(k) = K\bigl(r_i(k) - \xhat_i(k)\bigr) + \sum_{s \in \offsets} w_s \Bigl[\varphi\bigl(\delta_{i,s}(k)\bigr) + \varphi\bigl(\delta_{i,-s}(k)\bigr)\Bigr],
\label{codesign:eq:law}
\end{equation}
where $\delta_{i,s} = (\xhat_{i+s} - r_{i+s}) - (\xhat_i - r_i)$ is the difference of the formation errors of agent $i$ and its neighbour at offset $s$. Here $\xhat_i$ is the coordinator's estimate of $x_i$. For a graph that is not circulant, $w_s$ stands for agent $i$'s entries of the cyclic diagonals at the offsets $s$ and $-s$, one on each term, zero where there is no edge. The gain $K$ and the coefficients of $\varphi$ are the agents' design. The edge function $\varphi \colon \R^n \to \R^m$ is an odd polynomial of degree $\dphi$, a linear part $K_c\,\delta$ and an odd part of higher degree that stiffens the coupling of a neighbour far from its spacing. The worked example takes $K$ as the linear-quadratic gain of the sampled linear part of $f$ and $K_c = K/2$. Its edge function has $\dphi = 3$ with the cubic on the along-track component, its coefficient such that the cubic equals the linear term at a difference of 5 km. The specification is a tolerance on the position part of the formation error, in metres, after a transient of $k_0$ cycles,
\begin{equation}
\norm{C\bigl(x_i(k) - r_i(k)\bigr)} \le \tol \quad \text{for all } i \text{ and all } k \in \mathcal{K},
\label{codesign:eq:spec}
\end{equation}
where $\mathcal{K} = \bigcup_p \{s_p + k_0, \dots, e_p\}$, with $s_p$ and $e_p$ the first and last cycles of phase $p$ of the schedule ($\mathcal{K} = \{k_0, \dots, K\}$ for a schedule of one phase), over a mission of $K$ cycles. For the worked example $\tol = 50$ m with $k_0 h = 2\,860$ s, over the six orbits. The exact loop is designed so that the along-track separation of neighbours never falls below 5 km and no applied command exceeds $10^{-2}$ m/s$^2$ in Euclidean norm.

\subsection{Coordinator and encrypted computation}
\label{codesign:subsec:architecture}
\label{codesign:subsec:encrypted}

Each cycle has four steps. Every agent encrypts its measurement $y_i(k)$, the command $u_i(k-1)$ it applied over the last period and its references $r_i(k)$ and $r_i(k-1)$, and uploads them. The coordinator updates its retained estimate $\xhat(k) = (\xhat_1(k), \dots, \xhat_M(k))$ with the measurements, computes the commands of \cref{codesign:eq:law} and returns to agent $i$ a ciphertext that holds only $u_i(k)$. Agent $i$ decrypts it and applies the command. Between cycles the coordinator predicts the estimate forward with $\Phi_h$ and the uploaded commands, and every $T$ cycles it refreshes the retained estimate. The coordinator plays the role of the cloud in cloud-based encrypted control~\cite{schulzedarup2021encrypted}; the agents are its clients and compute nothing beyond encryption and decryption.

The fleet shares one key pair. Every agent holds the secret key, and the coordinator holds the public key, with which anyone can encrypt, and the evaluation keys, which allow multiplication and rotation of ciphertexts but not decryption. The coordinator is semi-honest: it executes the protocol faithfully and may use everything it sees to learn about the agents' data~\cite{schulzedarup2021encrypted,binfet2023towards}. It sees only ciphertexts, and the scheme's security guarantees that it learns nothing of the states, the measurements or the commands. It knows the fleet size, the graph and the shape of the computation. The model's coefficients, the gains and the references are the agents' and reach it encrypted, uploaded once before the first cycle (\cref{codesign:subsec:upload}). An agent learns only the command addressed to it, and the agents are independent parties who follow the protocol and do not collude. That all of them hold the same secret key is a limitation of the shared-key architecture, which we state and do not remove.

The scheme's decryptions are approximate, and over the lifetime of a key the errors of its decrypted values reveal the key~\cite{li2021security}. The standard practice therefore floods each decryption with Gaussian noise whose standard deviation $\sigfl$ is a multiple $c_{\mathrm{fl}}$, mandated at the security target, of a bound on the decrypted value's encryption error~\cite{li2022securing,cheon2024attacks}. Per released entry, after decoding, the bound is the library's own estimate $\hat{\varepsilon}$ of the released command's error (\cref{codesign:subsec:constraints}). In our loop the only decryptions are the agents' commands, so every applied command is $u_i(k) + \varepsilon_i(k) + \zeta_i(k)$, with $\varepsilon_i$ the encryption error of the released ciphertext and $\zeta_i$ the flooding noise. The multiple depends on one deployment parameter, the number of decryptions a key serves. The key schedule, which is outside this paper, sets it, and we take it as given at the library's default.

We write $\enc{z} = \Enc(z)$ for the ciphertext of a vector $z$ and $\Dec$ for decryption, with $\Dec(\enc{z}) = z + \varepsilon$ where $\varepsilon$ is the encryption error. One ciphertext holds $N/2$ real numbers, with $N$ the ring dimension, and the coordinator keeps the whole fleet in one, the $n$ components of agent $i$ in its own block of entries. The scheme offers three operations, all acting entry by entry: the addition $\enc{a} \ctadd \enc{b}$, the multiplication in its two forms, plaintext by ciphertext $c \ptmul \enc{a}$ and ciphertext by ciphertext $\enc{a} \ctmul \enc{b}$, and the cyclic rotation $\rot{s}(\enc{a})$ by $s$ entries. A ciphertext encrypted with $L$ levels admits $L$ consecutive multiplications, each spending one level, and must then be refreshed. The refresh is itself a computation of depth $\Lref$, so it returns $L - \Lref$ levels. It takes a time $\tref$ and changes every entry by at most $\beta$ except on a failure event of probability at most $\etaref$ per call (\cref{codesign:app:refresh}). Real numbers enter the scheme multiplied by a scaling factor $\Delta$, and $N$ is the smallest ring dimension the table of the Homomorphic Encryption Standard~\cite{HE-standard} allows at the security level for the modulus that $L$ levels at scale $\Delta$ require. We work at the library's 128-bit classical security setting throughout; \cref{codesign:app:terms} gives these terms in full.

\paragraph{The problem.}\label{codesign:subsec:problem} Given the agents, the graph, the specification \cref{codesign:eq:spec}, the mission length $K$ and the security level, choose the realization of the loop on ciphertexts and the parameters $(\Delta, L, N, T, h)$. The encrypted loop must meet the specification on an event of probability at least $1 - \eta$ with $\eta$ derived. Every cycle must complete within the period $h$, including the cycles that carry a refresh, and the security level and the mandated flooding must hold. The realization follows from the cost model of the scheme (\cref{codesign:sec:pipeline}) and the parameters from the specification (\cref{codesign:sec:guarantee,codesign:sec:design}).

\section{The Encrypted Pipeline}
\label{codesign:sec:pipeline}

The pipeline realizes one control cycle on ciphertexts. Each stage is written in the three operations of \cref{codesign:subsec:encrypted} and yields the quantities the design rule consumes: the depth it spends, its products and rotations, what it retains across cycles, and the error it contributes. The stages are those of the encrypted control loop as the tutorials draw it~\cite{schulzedarup2021encrypted,schluter2024code}, together with two that have no box there and a cost here, the packing of the fleet's uploads into one ciphertext and the release of each agent's command. One fact shapes every stage. Everything the coordinator multiplies by is an encrypted upload of the agents, so every product of the loop but the release's is a product of two ciphertexts. The middle band of \cref{codesign:fig:overview} draws the cycle with the level of each ciphertext. The tracing backend of \cref{codesign:subsec:tracer} reproduces the closed forms below operation by operation, and \cref{codesign:app:counts} gives each stage's operations and their counts for the worked example.

\subsection{Upload and state estimator}
\label{codesign:subsec:upload}
\label{codesign:subsec:estimator}

Each cycle agent $i$ encrypts the four vectors of \cref{codesign:subsec:architecture}, each zero outside its own block of $\nb$ of the ciphertext's $B$ entries; the references are needed because the coordinator retains the formation error rather than the state. The coordinator forms each fleet ciphertext by $M - 1$ additions, at no level and no rotation, so the stage costs each agent four encryptions and the coordinator $4(M - 1)$ additions per cycle, and the fleet ciphertexts enter the cycle with all $L$ levels. Every quantity the coordinator multiplies by is uploaded once before the first cycle (\cref{codesign:app:counts}).

The coordinator retains the fleet's formation error $\enc{e} = \enc{\xhat - r}$ across cycles, the quantity the law and the specification are written in; the estimate itself appears inside the cycle when the uploaded reference is added. The prediction evaluates the sampled dynamics on ciphertexts,
\begin{equation}
\enc{\xhat^{-}(k+1)} = \Phi_h\bigl(\enc{\xhat(k)}, \enc{u(k)}\bigr),
\label{codesign:eq:predict}
\end{equation}
where $\enc{u(k)}$ is the applied command uploaded at the next cycle and $\Phi_h$ is Heun's method, $s = 2$ evaluations of $f$ per period. One evaluation multiplies the components of its argument by the coefficient ciphertexts inside a balanced product tree, the coefficient counting as one more factor. With the coefficients encrypted every product is a product of two ciphertexts, so a field of degree $d_f$ spends $\lceil \log_2(d_f + 1) \rceil$ levels per evaluation and the prediction $\dpred = s\,\lceil \log_2(d_f + 1) \rceil$ (\cref{codesign:app:counts}). The update applies a constant gain to the innovation,
\begin{equation}
\enc{\xhat(k)} = \bigl(I_M \otimes (I - K_e C)\bigr)\,\enc{\xhat^{-}(k)} \ctadd \bigl(I_M \otimes K_e\bigr)\,\enc{y(k)},
\label{codesign:eq:update}
\end{equation}
with $K_e$ the steady-state Kalman gain of the sampled linear part of $f$, computed offline by the agents and uploaded as the cyclic diagonals of the two block-diagonal matrices. Each term is a sum over the nonzero cyclic diagonals of the block of the diagonal's ciphertext times a rotated copy, the diagonal method of Halevi and Shoup~\cite{halevi2014algorithms}, so the update spends one level whatever the dimension, and the two terms run in parallel. The retained chain therefore spends
\begin{equation}
\dret = \dpred + 1
\label{codesign:eq:dret}
\end{equation}
levels per cycle, and the only retained ciphertext is $\enc{e}$.

Two choices here follow from the cost model rather than from estimation theory. The prediction uses the uploaded applied command and not the released one, so the retained chain does not inherit the depth of the coupling, and the flooding noise in the applied command enters the estimator as a known input rather than an unknown disturbance. The gain is uploaded, so the update costs one level and the coordinator computes no gain on ciphertexts.

A third choice is where the field is evaluated. The polynomial $f$ is the expansion of the dynamics about one point, the origin of the fleet's frame, and its remainder grows with the distance from that point, in the worked example as the cube of the agent's along-track offset. Evaluating $f$ at the agents' positions is exact enough while the fleet is compact, 150 km to the farthest of sixteen agents, and it is the realization the mission runs. For a longer train the coordinator evaluates the same polynomial, with the same coefficient ciphertexts, at the estimate and at the agent's reference, and differences the two. The prediction becomes $\Phi_h(\enc{r + e}, \enc{u_{\mathrm{fb}}}) - \Phi_h(\enc{r}, 0) + \enc{\delta}$. Here $u_{\mathrm{fb}}$ is the applied command less the agent's own feedforward, and $\delta$ is the reference's one-period defect under the exact dynamics, which the agent computes in the clear from its own reference and uploads. The remainder then enters the disturbance only through its variation over the formation error, not through its value at the reference, the region of \cref{codesign:subsec:agents} becomes a box around each agent's reference, and the degree stays $d_f = 2$ at every fleet size. The cost is one more upload per cycle and a second evaluation of $f$ per stage, at no added depth. The released error falls, because the coefficient errors multiply the metre-scale formation error instead of the kilometre-scale reference (\cref{codesign:subsec:eval_loops}).

\subsection{Control law and release}
\label{codesign:subsec:law}
\label{codesign:subsec:release}

The law \cref{codesign:eq:law} has a tracking term and a coupling term, and the edge function has a linear part and an odd part of higher degree, $\varphi(\delta) = K_c\,\delta + \varphi_3(\delta)$ with $\varphi_3$ of degree $\dphi$. The linear part of the whole law is two block-diagonal products by the diagonal method. One is $(I_M \otimes K)\,\enc{e}$, with each agent's graph degree folded into its block. The other is $(I_M \otimes K_c)$ applied to the neighbour sum $\sum_{s \in \offsets} [\rot{s\nb}(\enc{e}) \ctadd \rot{-s\nb}(\enc{e})]$, which costs one rotation per offset and direction (\cref{codesign:app:terms}). The coupling is thus a sum over the nonzero cyclic diagonals of the adjacency matrix, as in the linear pipeline of~\cite{damera2026end}. The odd part is computed once per offset, and the term of the opposite direction follows by oddness at one rotation (\cref{codesign:app:counts}). The two parts run in parallel from $\enc{e}$, the linear one at one level and the odd one at $\lceil \log_2(\dphi + 1) \rceil$, so the stage spends $\lceil \log_2(\dphi + 1) \rceil$ levels, two for the cubic. The law is a branch from the retained error: its depth counts against the levels the estimate has at that point of the cycle but is not retained.

The coordinator multiplies $\enc{u}$ by the indicator vector $\iota_i$ of agent $i$'s command entries, one plaintext product per agent at one level. It adds one encryption of zero under the public key, the same ciphertext for every agent, which gives the flooding its mandated size at the first noise degree (\cref{codesign:app:counts}). It sends the sum to agent $i$, which decrypts it with the flooding and applies the command. When $M$ is a power of two the command chain of the cycle spends
\begin{equation}
\dcmd = \dret + \lceil \log_2(\dphi + 1) \rceil + 1
\label{codesign:eq:dcmd}
\end{equation}
levels from the retained error at the start of the cycle, the law's $\lceil \log_2(\dphi + 1) \rceil$ and the release's one above $\dret$. The released command's error, the encryption error of the released ciphertext plus the flooding noise, is the command part of the per-cycle error of \cref{codesign:sec:guarantee}.

\subsection{Refresh, cost model and tracing backend}
\label{codesign:subsec:refresh}
\label{codesign:subsec:tracer}

After the release of every $T$th cycle, the first interval on the fresh ciphertext being longer, the coordinator refreshes $\enc{e(k)}$ and continues with $L - \Lref$ levels, where $\Lref$ is one larger for the two-iteration refresh. Between refreshes the retained error spends $\dret$ levels per cycle and the command branch needs $\dcmd - \dret$ levels above it, so $T$ cycles fit if
\begin{equation}
T\,\dret + (\dcmd - \dret) \le L - \Lref .
\label{codesign:eq:fit}
\end{equation}
The refresh blocks the estimate while it runs, so a cycle with a refresh lasts $\tcyc + \tref$, where $\tcyc$ is the time of the prediction, the update, the law and the release, and the period must accommodate it (\cref{codesign:eq:period}).

For the worked example on the ring of sixteen agents the pipeline spends $\dret = 5$ and $\dcmd = 8$ levels with 52 ciphertext products, 16 plaintext products and 38 rotations per cycle (\cref{codesign:prop:cost,codesign:prop:counts}). At $L = 39$ with $\Lref = 21$ for the two-iteration refresh it refreshes every $T = 3$ cycles, with equality in \cref{codesign:eq:fit}.

The coordinator's computation is a fixed sequence of operations that does not depend on the values in the ciphertexts. A tracing backend therefore runs the pipeline's code on unencrypted arrays and records every operation with its kind, the level and noise degree of its result and the rotation index, in the library's order and under its level rules. On every circuit tested the record equals the library's own record operation by operation (\cref{codesign:sec:evaluation}). The trace of one cycle gives the quantities of \cref{codesign:prop:cost,codesign:prop:counts} for any pipeline composed from these stages, the levels at which each retained ciphertext would run out, and the time of a cycle at a ring dimension $N$ from the measured times of the operations. The design rule of \cref{codesign:sec:design} reads $\dret$, $\dcmd$ and $\tcyc(N)$ from this trace.

\section{Closed-Loop Guarantee}
\label{codesign:sec:guarantee}

The guarantee compares two loops driven by the same disturbances and noise. The exact loop is the plant of \cref{codesign:eq:sampled} with the coordinator computing the estimator and the law of \cref{codesign:eq:law} in exact arithmetic. The encrypted loop is the same plant with the coordinator computing on ciphertexts as in \cref{codesign:sec:pipeline}, the commands released with their encryption error and flooding noise, and the estimate refreshed every $T$ cycles. Let $z(k) = (x(k) - r(k), \xhat(k) - r(k)) \in \R^{2Mn}$ be the joint state of the exact loop in formation-error coordinates and $z_{\mathrm{enc}}(k)$ that of the encrypted loop. Let $G_k(z, \nu)$ be the map of the exact loop when the agents apply the coordinator's commands perturbed by $\nu \in \R^{Mm}$ and upload the commands they applied, so that $z(k+1) = G_k(z(k), 0)$; it depends on $k$ through the references, the feedforward and the disturbances. Write $G_k(z) = G_k(z, 0)$, $J_k(z)$ for the Jacobian of $G_k$ in $z$ at $\nu = 0$ and $B_k(z)$ for its Jacobian in $\nu$.

\subsection{Assumptions}
\label{codesign:subsec:assumptions}

\begin{assumption}[Class and region]
\label{codesign:ass:class}
The field $f$ is polynomial of degree $d_f$, and the region $\region$ of \cref{codesign:subsec:agents} is compact and convex; $\varphi$ is an odd polynomial of degree $\dphi$. The disturbances satisfy $\norm{w_i(k)} \le \bar{w}$ and $\norm{v_i(k)} \le \bar{v}$.
\end{assumption}

\begin{assumption}[Encryption errors]
\label{codesign:ass:errors}
Let $\mathcal{F}_k$ be everything drawn before cycle $k$: the key, the coefficient ciphertexts uploaded once before the first cycle, the noise and the disturbances, and the randomness of the cycles and refreshes before $k$. The joint states $z(k)$ and $z_{\mathrm{enc}}(k)$ are determined by it. There are constants $\abar$, $\beta$, $\sigfl$, computed in \cref{codesign:app:percycle,codesign:app:refresh}, and four events fixed for the mission, with complements of probability at most $\etakey$, $\eta_{\mathrm{coef}}$, $\eta_v$ and $\eta_\zeta$. They are the key event, the complement of the key's failure event of \cref{codesign:def:keyevent}, and the coefficient event that every coefficient entry's error is within $z_{\mathrm{c}}$ deviations. They are also the fiber event that every encryption of the mission has fiber energy at most $\bar{A}_v$, and the command-box event that every released entry's flooding is within $z_\zeta \sigfl$. On their intersection, three properties hold. Take any cycle $k$ and any joint state $z$ in $\errbox \times \errbox$, and run the cycle on ciphertexts from $z$ under the same uploads, with the flooding $\zeta(k)$ of the cycle's released commands applied and uploaded. The result differs from $G_k(z, \zeta(k))$ by more than $\abar$ in the metric of \cref{codesign:ass:contraction} with conditional probability at most $\etacyc$ given $\mathcal{F}_k$. For every refresh and every retained estimate in the region, the refresh changes some entry of it by more than $\beta$ (in encoded units, \cref{codesign:app:refresh}) with conditional probability at most $\etaref$ given the past. The flooding $\zeta(k) \in \R^{Mm}$ added to the released commands of cycle $k$ is Gaussian with independent entries of standard deviation $\sigfl$, independent of $\mathcal{F}_k$ and of the cycle's other randomness. The constants are derived under three models of the scheme's randomness, stated in \cref{codesign:app:refresh}.
\end{assumption}

The constant $\abar$ collects the arithmetic error of the retained chain and of the released command. Both grow with the size of the operands and are therefore bounded on the region. The flooding is not bounded, and the theorem keeps it as the Gaussian it is (\cref{codesign:rem:standard}).

\begin{assumption}[Contraction certificate for the exact loop]
\label{codesign:ass:contraction}
There are a norm $\norm{\cdot}_P$ on $\R^{2Mn}$, an integer $\ell \ge 1$ and constants $\gamma < 1$, $c \ge 1$ and $L_2 \ge 0$ with the following properties on $\errbox \times \errbox$, where $\errbox \subset \R^{Mn}$ is a compact convex box of formation errors whose states lie in $\region$ (for the worked example, 360 m per position axis and 1 m/s per velocity axis around the references). Each one-cycle map is twice continuously differentiable, and its expansion in the command, $G_k(z, \nu) = G_k(z) + B_k(z)\,\nu + n_k(z, \nu)$, satisfies, for all $z, z'$ in the region and every $\nu$ whose entries are at most $\bar{\nu}$ in magnitude, $\norm{(B_k(z) - B_k(z'))\nu}_P \le L_B\,\sqrt{m}\,\bar{\nu}\,\norm{z - z'}_P$ and $\norm{n_k(z, \nu)}_P \le L_\nu\,m\,\sqrt{M}\,\bar{\nu}^2$. For every $k$ the $\ell$-cycle map satisfies $\norm{G_{k+\ell-1} \circ \dots \circ G_k(z) - G_{k+\ell-1} \circ \dots \circ G_k(z')}_P \le \gamma \norm{z - z'}_P$ for all $z, z'$ in the region, and each one-cycle map is $c$-Lipschitz there. The remainder of each one-cycle map's linearization satisfies $\norm{G_k(z + d) - G_k(z) - J_k(z)\,d}_P \le L_2 \norm{d}_P^2$ whenever $z$ and $z + d$ lie in the region. The region with a margin $\rho_{\errbox}$ is invariant for the exact loop under \cref{codesign:ass:class}.
\end{assumption}

The certificate is computed per period $h$ for the chosen agents in closed form, not sampled. Its norm weights each agent's formation error by the Riccati matrix $P_r$ of the linear quadratic tracking design that produced the gain and its estimation error $\tilde{x}_i = x_i - \xhat_i$ by a Lyapunov matrix $P_e$ of the estimator's error dynamics (\cref{codesign:app:percycle}). It contracts the linear part of the coupled loop for every graph whose largest Laplacian eigenvalue times the coupling weight is at most a limit computed in this norm. The Jacobian's variation over the region and $L_2$ are bounded through the terms of $f$ and $\varphi$ (\cref{codesign:subsec:eval_certificate,codesign:app:proofs}). The certified map is the plant's own sampled flow under the held commands, with the coordinator's Heun step in the estimate rows (\cref{codesign:app:proofs}).

\subsection{Deviation bound}
\label{codesign:subsec:tube}

\begin{lemma}[The encrypted loop as a perturbed linearization]
\label{codesign:lem:additive}
Under \cref{codesign:ass:class,codesign:ass:errors,codesign:ass:contraction}, on the intersection of the events of \cref{codesign:ass:errors} over the mission, the deviation $d(k) = z_{\mathrm{enc}}(k) - z(k)$ satisfies
\begin{equation}
d(k+1) = J_k\,d(k) + B_k\,\zeta(k) + a(k) + \rho(k),
\label{codesign:eq:dev}
\end{equation}
with $J_k = J_k(z(k))$, $B_k = B_k(z(k))$ and
\begin{align*}
\norm{a(k)}_P &\le \abar + \abar_\nu\bigl(\norm{d(k)}_P\bigr) + \beta'\,\mathbf{1}[\text{a refresh follows cycle } k], \\
\norm{\rho(k)}_P &\le L_2\,\norm{d(k)}_P^2,
\end{align*}
where $\abar_\nu(d) = L_B\,\sqrt{m}\,z_\zeta\sigfl\,d + L_\nu\,m\,\sqrt{M}\,(z_\zeta\sigfl)^2$ and $\beta'$ is the refresh's image through the prediction (\cref{codesign:lem:refreshpred}), at every $k$ at which $z_{\mathrm{enc}}(k) \in \errbox \times \errbox$.
\end{lemma}

\begin{theorem}[Deviation from the exact loop]
\label{codesign:thm:tube}
Let \cref{codesign:ass:class,codesign:ass:errors,codesign:ass:contraction} hold with $T \ge \ell$, and write $c_i = \sum_{j=0}^{i-1} c^j$ for $i \ge 0$. For an input bound $a \ge 0$ define the deterministic radii
\begin{align}
r_{\mathrm{b}}(a, \beta, T) &= \frac{c_\ell\,a}{1 - \gamma} + \frac{c_\ell\,\beta}{1 - \gamma^{\lfloor T/\ell \rfloor}}, \label{codesign:eq:rb} \\
\rstar(a, \beta, T) &= c^{\ell-1}\,r_{\mathrm{b}}(a, \beta, T) + c_{\ell-1}\,(a + \beta) . \label{codesign:eq:rstar}
\end{align}
Let $\Sigma(k)$ be the covariance of the flooding's response along the exact trajectory,
\begin{equation}
\Sigma(k+1) = J_k\,\Sigma(k)\,J_k^{\top} + \sigfl^2\,B_k B_k^{\top}, \qquad B_k = B_k(z(k)), \quad \Sigma(0) = 0 ,
\label{codesign:eq:cov}
\end{equation}
let $\ghat(\etag)$ bound $\norm{g}_P$ for $g \sim \mathcal{N}(0, \Sigma(k))$ at every cycle of the mission except with probability $\etag$, and let
\begin{equation}
\rg(\etapos) = q(\etapos)\,\max_{i,\,k}\,\sqrt{\lambda_{\max}\bigl(\Sigma_{i}(k)\bigr)} ,
\label{codesign:eq:rg}
\end{equation}
where $\Sigma_i(k)$ is the block of $\Sigma(k)$ on agent $i$'s position error and $q(\etapos)^2$ is an upper bound on the quantile of the chi-square law with three degrees of freedom at the level $\etapos/(MK)$ ($q = 7.01$ at the mission; \cref{codesign:app:proofs}). Suppose the equation $\dhat = \ghat + \rstar(\abar + \abar_\nu(\dhat) + L_2 \dhat^2, \beta', T)$ has a solution with $\dhat \le \rho_{\errbox}$, and take the smallest. If $z_{\mathrm{enc}}(0) = z(0) \in \errbox \times \errbox$, then on an event of probability at least $1 - \eta$ with
\begin{equation}
\eta = \etakey + R\,\etaref + \eta_{\mathrm{coef}} + K\,\etacyc + \eta_v + \eta_\zeta + \etag + \etapos ,
\label{codesign:eq:eta}
\end{equation}
the encrypted loop stays in $\errbox \times \errbox$, $\norm{d(k)}_P \le \dhat$ for all $0 \le k \le K$, and every agent's position deviates from the exact loop's by at most
\begin{equation}
\rpos = \rg(\etapos) + \kappa_P\,\rstar\bigl(\abar + \abar_\nu(\dhat) + L_2 \dhat^2, \beta', T\bigr) ,
\label{codesign:eq:rpos}
\end{equation}
where $R \le \lceil K/T \rceil$ is the number of refreshes, $\kappa_P$ bounds the Euclidean norm of any agent's position part of $z$ by $\norm{z}_P$, and $\etakey$, $\eta_{\mathrm{coef}}$, $\eta_v$ and $\eta_\zeta$ are the probabilities of the four mission events of \cref{codesign:ass:errors}.
\end{theorem}

We call the band of radius $\rpos$ around the exact loop the tube. The proofs of this section are in \cref{codesign:app:proofs}, where \cref{codesign:cor:envelope} gives the radius at each cycle. The deviation splits into two parts. The flooding's response $g$ is Gaussian with the covariance of \cref{codesign:eq:cov}, because the linearization along the exact trajectory is linear in the noise. The remainder $r$ is driven by the arithmetic, the refresh and the second-order term, and the contraction bounds it as a recursion with an additive input. The position radius takes the Gaussian part from its covariance rather than through the metric.

\begin{corollary}[The specification]
\label{codesign:cor:spec}
Under the conditions of \cref{codesign:thm:tube}, if the exact loop satisfies $\norm{C(x_i(k) - r_i(k))} \le \tol - \rpos$ for $k \in \mathcal{K}$, then the encrypted loop satisfies the specification \cref{codesign:eq:spec} on the same event.
\end{corollary}

\paragraph{The computation delay.} The coordinator's feedback for cycle $k$ reaches agent $i$ at $kh + \tau_k$, after the cycle's computation and the decryption, with $0 \le \tau_k < h$; the agent applies its own feedforward on time and flies the previous feedback until then. The stale feedback $s_i(k) = u_i(k-1) + \nu_i(k-1)$, with $\nu_i$ the released command's error, is an added state, $z_i = (e_{x,i}, e_{\hat{x},i}, s_i) \in \R^{2n + m}$. The plant's sampled flow then holds $s_i(k)$ on $[kh, kh + \tau_k)$ and the new feedback on $[kh + \tau_k, (k+1)h)$, which splits its input sensitivity into $\Psi_u(\tau) = \int_\tau^h \Psi(h, \sigma)B\,d\sigma$ for the new command and $\Psi_s(\tau) = \int_0^\tau \Psi(h, \sigma)B\,d\sigma$ for the stale one, with $\Psi(h, \sigma)$ the state transition of the linearized flow from $\sigma$ to $h$; their sum is the undelayed sensitivity. The coordinator's prediction is unchanged and applies the uploaded command over the whole period.

\begin{proposition}[The certificate with the delay]
\label{codesign:prop:delay}
Let $\Theta \subset [0, h)$ be a set of delays and $P = I_M \otimes P_a$ with $P_a$ symmetric positive definite on $\R^{2n + m}$. Let the delayed loop's one-cycle maps $G^{\tau}_k$ (the plant's own sampled flow with the delay $\tau$, the coordinator's Heun step in the estimate rows) satisfy, for every $\tau \in \Theta$ and in $\norm{\cdot}_P$ on the region extended by the stale feedback, the conditions (i) to (iv) of \cref{codesign:app:proofs}: the contraction of \cref{codesign:ass:contraction} at the rate $\gamma_\Theta$ with $\ell = c = 1$, its remainder bound, its command expansion and its margin. Then \cref{codesign:thm:tube} holds for the delayed loop under every delay sequence $(\tau_k) \subset \Theta$, with $\gamma = \gamma_\Theta$, $\ell = c = 1$, $J_k = J^{\tau_k}_k(z(k))$ and $B_k = B^{\tau_k}_k(z(k))$.
\end{proposition}

The norm has to be a full quadratic form on $(e_x, e_{\hat{x}}, s)$: the delay-unaware prediction couples the estimate's error to the lag of the stale command behind the current one, both ways. On the ring of sixteen at 65 s, in that norm and on the plant's own sampled map, $\gamma_\Theta = 0.9891$ for every delay sequence in $\Theta = [14, 24]$ s. In the interval's norm the certified radius is $2.36$ m at 24 s, and the exact loop's held error plus the radius is 46.2 m against 50 m. \Cref{codesign:app:config} gives the fixed delays.

\begin{remark}[What the theorem contributes]
\label{codesign:rem:standard}
The deterministic part of \cref{codesign:eq:rpos} is the standard bound for a contracting recursion with an additive input. For linear agents with a steady-state gain the certificate is global with $\ell = 1$ and $L_2 = 0$ (\cref{codesign:cor:linear} in \cref{codesign:app:proofs}), and without the flooding the radius recovers the bound of~\cite{damera2026end}. What the theorem adds is twofold. Every constant on its right side is computed. The certificate gives $\gamma$, $c$, $\ell$, $L_2$ and the Jacobians from the agents and the graph. The pipeline's trace and the library's circuits give $\abar$, $\beta$, $\sigfl$, $\etaref$ and $\etacyc$, and the depth gives $T$. And the flooding, the one error that is Gaussian by mandate rather than bounded by arithmetic, enters through the covariance recursion rather than as a bound per cycle. That is what allows the bound to be inverted in \cref{codesign:sec:design}.
\end{remark}

\section{The Design Rule}
\label{codesign:sec:design}

The design rule inverts the tube of \cref{codesign:thm:tube}, which depends on the parameters in both directions. We first write the constraints, then the algorithm that enumerates them, then the theorem on its output.

\subsection{The constraints}
\label{codesign:subsec:constraints}

A configuration is a tuple $(\Delta, L, N, T, h)$: the scaling factor, the levels, the ring dimension, the refresh interval and the control period. The scaling factor and the levels are chosen from the finite menu the library supports, and the remaining three follow from them and from the loop, the period being the least $h(\Delta, L)$ of \cref{codesign:eq:period}.

\emph{Security.} At 128-bit security the ring dimension is the smallest $N$ in the table of the standard whose modulus admits $L$ levels at scale $\Delta$ together with the refresh:
\begin{equation}
N = N_{128}(\Delta, L) .
\label{codesign:eq:ring}
\end{equation}

\emph{Time.} The cycle time $\tcyc(N, L)$ is read from the trace with the operation times measured at $N$ and at the levels the cycle visits. The refresh time $\tref(N, L)$ is measured once per configuration at the fleet's number of occupied entries. The refresh runs after the release of the interval's last cycle, so the period must carry the larger of the interval's first cycle and its last cycle with the refresh, with a margin $\mu$ against the variation of the measured times, on a grid of 5 s:
\begin{equation}
h \ge (1 + \mu)\,\max\bigl(t_1(N, L),\ t_T(N, L) + \tref(N, L)\bigr), \qquad h \in 5\,\mathbb{N}\ \text{s},
\label{codesign:eq:period}
\end{equation}
where $t_p(N, L)$ is the trace's time of the coordinator's cycle at position $p$ of the refresh interval, at the levels it visits; the rule takes the least such $h$, $h(\Delta, L) = 5\lceil (1 + \mu)\max(t_1, t_T + \tref)/5 \rceil$ s, with $\mu = 0.1$.

\emph{Depth.} The refresh interval fits the levels, \cref{codesign:eq:fit}, which gives the most cycles between refreshes,
\begin{equation}
T \le T_{\max}(L) = \Bigl\lfloor \frac{L - \Lref(\Delta, L) - (\dcmd - \dret)}{\dret} \Bigr\rfloor .
\label{codesign:eq:tmax}
\end{equation}

\emph{Tolerance and probability.} The exact loop's formation error on its own, $\rho_{\mathrm{ex}}(h)$, and the tube together stay within the tolerance, \cref{codesign:cor:spec}. The guarantee's failure probability $\eta$ of \cref{codesign:eq:eta}, computed from $\etaref(N)$, $\etacyc$ and $\etakey$ over the mission length $K$, stays within a target $\eta^\ast$ the designer sets. Both bound the refresh interval from below, because each refresh adds $\beta$ to the tube, whose series in \cref{codesign:eq:rstar} falls with $T$, and $\etaref$ to $\eta$:
\begin{equation}
\begin{aligned}
T \ge T_{\min}(h, \Delta, N) = \min\{T \ge \ell :\;& \dhat(T) \le \rho_{\errbox}, \\ & \rho_{\mathrm{ex}}(h) + \rpos(T) \le \tol, \\ & \eta(T) \le \eta^\ast\},
\end{aligned}
\label{codesign:eq:tmin}
\end{equation}
where $\abar$, $\beta$ and $\sigfl$ are the error constants at $(\Delta, N)$, and $\gamma$, $c$, $\ell$, $L_2$ and the Jacobians inside $\dhat$ and $\rpos$ are the certificate's at the period $h$. The target is a cryptographic lower bound on the interval: at the worked example it is the one that binds (\cref{codesign:subsec:interpretation}).

\emph{Flooding.} The mandated flooding per released entry is $\sigfl = c_{\mathrm{fl}}\,\hat{\varepsilon}(\Delta, N)\sqrt{2B}/\Delta$. Here $\hat{\varepsilon}(\Delta, N)$ is predicted from the circuits and measured by the library's estimation phase (\cref{codesign:app:percycle}). The multiple $c_{\mathrm{fl}} = \sqrt{12 n_{\mathrm{d}}}\,2^{s/2}$ is set by the library's rule from the security target $s$ and the number $n_{\mathrm{d}}$ of decryptions the key serves (\cref{codesign:subsec:architecture}), about $2^{16.8}$ at $s = 30$ and $n_{\mathrm{d}} = 1$. The flooding enters $\rg$ through \cref{codesign:eq:cov} and so appears inside \cref{codesign:eq:tmin}.

The constraints run in a cycle. More levels raise $T_{\max}$ linearly, \cref{codesign:eq:tmax}, but raise $N$, \cref{codesign:eq:ring}, hence the period, \cref{codesign:eq:period}. A longer period contracts more per cycle and needs fewer refreshes over the mission, but raises the exact loop's held error $\rho_{\mathrm{ex}}(h)$, which shrinks the margin $T_{\min}$ must fit, \cref{codesign:eq:tmin} (\cref{codesign:tab:certificate} gives both across the menu). The feasible refresh intervals at a given $(\Delta, L)$ form the interval $[T_{\min}, T_{\max}]$, and the design is a pair $(\Delta, L)$ at which it is not empty. A finer scale lowers $\abar$ and $\beta$ and the flooding, and costs modulus, hence ring, hence time. A coarser scale holds more levels at the same ring, hence a longer interval and a smaller $\eta$, at a larger tube.

\subsection{The algorithm}
\label{codesign:subsec:algorithm}

\begin{algorithm}[t]
\caption{Control and cryptosystem co-design}
\label{codesign:alg:codesign}
\begin{algorithmic}[1]
\Require agents $(f, B, C)$ on $\region$, integrator $\Phi$ with $s$ stages, graph $\graph$ with offsets $\offsets$, edge function $\varphi$, gain schedule, tolerance $\tol$, mission $K$, failure-probability target $\eta^\ast$, security level, the library's menu of $(\Delta, L)$
\Ensure a configuration $(\Delta, L, N, T, h)$ with its constants $(\abar, \beta, \sigfl, \gamma, c, \ell, \eta)$, $\eta \le \eta^\ast$, or infeasible with the binding constraint
\State trace one cycle of the pipeline: $\dret$, $\dcmd$, the operation counts \Comment{\cref{codesign:subsec:tracer}}
\ForAll{$(\Delta, L)$ in the menu, in the order of the period $h(\Delta, L)$ of \cref{codesign:eq:period} at $N = N_{128}(\Delta, L)$}
    \State $N \gets N_{128}(\Delta, L)$; $\tref, \tcyc \gets$ times at $(N, L)$; $h \gets 5\lceil (1 + \mu)\max(t_1, t_T + \tref)/5\rceil$ \Comment{\cref{codesign:eq:ring,codesign:eq:period}}
    \State $(\gamma, c, \ell, L_2, J_k) \gets$ certificate of the exact loop at period $h$ \Comment{\cref{codesign:ass:contraction}}
    \State \textbf{if} no certificate at $h$, record $h$ as binding and \textbf{continue}
    \State $\beta, \etaref \gets$ refresh error at $(\Delta, N)$; $\abar, \sigfl, \etacyc \gets$ per-cycle error and flooding at $(\Delta, N)$ \Comment{\cref{codesign:app:refresh,codesign:app:percycle}}
    \State $T_{\max} \gets$ \cref{codesign:eq:tmax}; $T_{\min} \gets$ \cref{codesign:eq:tmin}
    \If{$T_{\min} \le T_{\max}$}
        \State $T \gets T_{\max}$; $\eta \gets$ \cref{codesign:eq:eta}
        \State \Return $(\Delta, L, N, T, h)$ with $(\abar, \beta, \sigfl, \gamma, c, \ell, \eta)$
    \Else
        \State record which of \cref{codesign:eq:tmax,codesign:eq:tmin} binds
    \EndIf
\EndFor
\State \Return infeasible, with the binding constraint of each $(\Delta, L)$
\end{algorithmic}
\end{algorithm}

\Cref{codesign:alg:codesign} enumerates the menu in the order of the period each pair yields, computed from the table and the trace before the loop, and returns the first feasible configuration, which therefore has the least period. Within a feasible pair it takes the largest refresh interval, which gives the fewest refreshes over the mission and the smallest $\eta$. The certificate, recomputed at every period because the sampled loop changes with $h$, is the algorithm's one step on the control side; the trace, the table and the error constants are its steps on the cryptographic side. The two sides exchange four quantities. The control side hands over the fewest cycles between refreshes, and the cryptographic side hands back the period, the error constants and the fewest cycles its failure probability allows. \Cref{codesign:fig:feasible} draws the feasible set.

\begin{theorem}[The design]
\label{codesign:thm:design}
Let the certificate step return the constants of \cref{codesign:ass:contraction} whenever its test passes on the region at the given period, and let the error steps return the constants of \cref{codesign:ass:errors}. Suppose \cref{codesign:alg:codesign} returns a configuration. Then the encrypted loop realized with it meets the specification \cref{codesign:eq:spec} on an event of probability at least $1 - \eta$ with $\eta \le \eta^\ast$. The cycle that includes the refresh completes within $h$ whenever the measured times stay within the margin $\mu$ of their values at the design. No configuration in the menu with a shorter period passes the certificate test and satisfies \cref{codesign:eq:ring,codesign:eq:period,codesign:eq:tmax,codesign:eq:tmin}. If the algorithm returns infeasible, no configuration in the menu does.
\end{theorem}

The proof is in \cref{codesign:app:proofs}. The first claim is \cref{codesign:cor:spec} with the returned constants, the second is \cref{codesign:eq:period}, and the third and fourth are the order and the exhaustiveness of the enumeration over the menu with the certificate test as given.

The flooding does not change the algorithm; it changes what the tolerance has to absorb. Its share of the tube is $\rg$, linear in $\sigfl$ and set by the loop's response to white noise on the commands, \cref{codesign:eq:cov}. The flooding also enters through $\dhat$ in the remainder and through $\abar_\nu$, which add 0.09 m at the mission. The loop's remaining margin, $\tol - \rho_{\mathrm{ex}}(h) - \kappa_P\,\rstar$ computed without the flooding, absorbs it. If the margin is smaller than $\rg$, no refresh interval satisfies \cref{codesign:eq:tmin} and the algorithm reports the tolerance as binding. There are three remedies. A finer scale lowers $\hat{\varepsilon}$ and so the flooding itself. A realization with a smaller released error lowers it too, since $\hat{\varepsilon}$ is set by the largest arithmetic error over the fleet (\cref{codesign:app:percycle}). A controller with a smaller response to command noise lowers $\rg$ at the same flooding. The flooding is applied at decryption, after every computation of the loop, yet it is a term of the closed-loop bound with a coefficient the loop sets, and the number of decryptions the key serves enters that term through $\sqrt{n_{\mathrm{d}}}$. For the worked example it is the largest part of the radius on the ring of sixteen and decides the torus (\cref{codesign:subsec:eval_certificate}).

\section{Evaluation}
\label{codesign:sec:evaluation}

\begin{figure}[!t]
\centering
\includegraphics[width=\columnwidth]{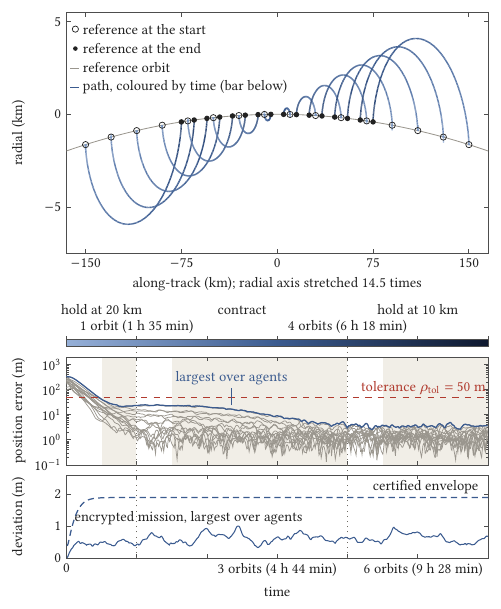}
\caption{The mission on the ring of sixteen: paths (top), position errors (middle) and the deviation from the exact loop (bottom).}
\label{codesign:fig:loops}
\end{figure}

The experiments answer four questions. Does the tracing backend reproduce the library's operation record, so that the counts of \cref{codesign:prop:cost} are exact? Do the derived refresh error and its failure probability hold for the library at standard security? Does the exact loop of the worked example admit the certificate of \cref{codesign:ass:contraction} on a region large enough for its mission? And does the encrypted loop, realized with the configuration that \cref{codesign:alg:codesign} returns, meet the specification within the designed period across coordination graphs? Every configuration is reported as the library prints it (\cref{codesign:tab:config}). The pass criteria of each experiment are fixed before its numbers (\cref{codesign:app:experiments}).

The worked example is the satellite train of \cref{codesign:subsec:agents}. \Cref{codesign:eq:period} gives $h = 65$ s at $N = 131\,072$ with $\mu = 0.1$, so the mission is $K = 524$ cycles and the transient $k_0 = 44$ cycles. In exact arithmetic the loop meets its specification (\cref{codesign:tab:certificate}), with a largest applied command of $5.5 \times 10^{-3}$ m/s$^2$ in Euclidean norm and a smallest neighbour separation of 9.99 km after the contraction. The package runs the pipeline of \cref{codesign:sec:pipeline} on three interchangeable backends over the same code, the library, the tracing backend and plain arithmetic. The plain backend only checks that the encrypted computation is correct and is not a scenario of the paper, since the coordinator never sees plaintext. The mission on the ring of sixteen ran encrypted at 128 bits for its full length; the other full missions in this section are emulated, the loop in plain arithmetic with the three error terms drawn from the models of \cref{codesign:app:refresh,codesign:app:percycle}, calibrated on that mission. Every configuration's cycle and refresh times are measured at 128 bits in encrypted runs of six to twelve cycles.

\subsection{Trace and refresh validation}
\label{codesign:subsec:eval_trace}
\label{codesign:subsec:eval_refresh}

The record matched on every circuit tested, 44 of 44 tests over 5\,138 recorded operations. The pipeline's own traced cycle matches its executed circuit over 6\,172 operations with no mismatch. At the 128-bit configuration the same counts, priced with the measured operation times of \cref{codesign:tab:optimes}, predict the coordinator's cycle after a refresh within 18 percent of the measured one. That holds on the ring at 16, 32, 64 and 128 agents and on the five other families at sixteen (\cref{codesign:subsec:eval_loops}). That is the check the counts admit at that setting.

The bound $\beta$ and the probability $\etaref$ of \cref{codesign:app:refresh} are computed from the library's configuration, and all 18 checks on sampled refreshes passed in its three configurations. Their values at the mission's configuration and the check on the pipeline's own retained ciphertext are in \cref{codesign:app:refresh}.

\subsection{Certificate and closed-loop results}
\label{codesign:subsec:eval_certificate}
\label{codesign:subsec:eval_loops}

\begin{table*}[t]
\caption{The ring at 32, 64 and 128 agents at 128 bits, with $\rpos$ split into its flooding (G), arithmetic (A) and refresh (R) parts.}
\label{codesign:tab:sizes}
\centering
\footnotesize
\setlength{\tabcolsep}{3.5pt}
\begin{tabular}{lcccccccc}
\toprule
$M$, $h$ [s] & $\hat{\varepsilon}$ [bits] & $\rpos$ (G, A, R) [m] & $\rho_{\mathrm{ex}}$ 0 / 22 s [m] & $\rho_{\mathrm{ex}} + \rpos$ [m] & $\gamma$, $\rho_{\errbox}$ & dev. / env. [m] & refresh cycles [s] & after refresh [s] \\
\midrule
32, 65 (run) & 23.3 & 0.41 (0.23, 0.04, 0.14) & 9.98 / 10.42 & 10.4 & 0.965, 0.743 & 0.069 / 0.28 & 63.3, 66.3 & 26.1 \\
32, 70 (rule) & 23.3 & 0.41 (0.24, 0.04, 0.13) & 10.49 / 10.52 & 10.9 & 0.963, 0.720 & & & \\
64, 65 (run) & 23.2 & 1.09 (0.32, 0.32, 0.46) & 12.73 / 13.13 & 13.8 & 0.967, 0.736 & 0.108 / 0.69 & 64.0, 64.7 & 26.1, 26.3 \\
64, 70 (rule) & 23.2 & 1.06 (0.33, 0.30, 0.44) & 12.63 / 12.79 & 13.7 & 0.965, 0.708 & & & \\
128, 75 (rule and run) & 24.4 & 9.99 (1.20, 2.97, 5.82) & 19.48 / 19.51 & 29.5 & 0.960, 81.3 m & 0.601 / 13.5 & 70.8, 72.6 & 26.1, 26.3 \\
\bottomrule
\end{tabular}
\end{table*}

All five criteria of \cref{codesign:app:experiments} pass. On the ring of sixteen the certificate holds on the plant's own sampled map at every period of the menu with $\ell = 1$ (\cref{codesign:tab:certificate}). Across the families the certificate is one number per period. It holds whenever the coupling weight times the graph's largest Laplacian eigenvalue is at most 10.0 at 65 s, 13.0 at 20 s and 5.2 at 130 s, so the certified weight is four times larger on the ring than on the complete graph. At the mission's configuration the tube is $\rpos = 1.91$ m against the exact loop's 5.86 m margin to the tolerance, and 1.52 m of it is the flooding (\cref{codesign:app:percycle}). The flooding taken instead at its per-cycle quantile through $\kappa_P/(1 - \gamma)$ would give about 610 m. Across the families (\cref{codesign:tab:families}) each tube is computed at the family's own released-error estimate, measured at 128 bits in a six-cycle run at its instance weight (\cref{codesign:tab:familiesfull}). At 65 s every family but the complete graph is certified at weight 0.25, and the 2-nearest ring and the hypercube also at their instance weight (\cref{codesign:app:experiments}).

On the ring of sixteen the full mission ran encrypted at the configuration of \cref{codesign:tab:config}, 524 cycles in 7.9 hours of wall time on the coordinator's machine. The deviation, the formation error and the refresh count pass on every cycle, and the time on 523 of the 524 (\cref{codesign:app:experiments}). The encrypted loop's largest formation error after the transients was 43.8 m against the tolerance of 50 m. Its largest deviation from the exact loop was 1.01 m, at cycle 213, with a median of 0.59 m over the mission, inside the certified radius. \Cref{codesign:fig:loops} shows the mission, the paths contracting and the position errors under the tolerance. The deviation stays under the per-cycle envelope of \cref{codesign:cor:envelope} at every cycle, peaks in the contraction phase and shows no step at the refresh cycles. At sixteen agents the other five families ran six cycles each at the same configuration, and the family moves the cycle time by at most an eighth, the complete graph being the slowest (\cref{codesign:tab:familiesfull}). On the ring at 32, 64 and 128 agents the field is evaluated about each agent's reference (\cref{codesign:subsec:estimator}), and the loop runs with the coordinator's feedback delayed by 22 s. \Cref{codesign:tab:sizes} gives, at each size, the rule's period from the measured costs, the twelve-cycle encrypted run at 128 bits and the certificate. Its cycle times at 64 and 128 agents are composed from stage times measured on an unloaded machine. The released error's estimate stays at $2^{23.2}$ to $2^{23.3}$ at 32 and 64 agents, below the sixteen-agent ring's $2^{25.9}$. The probability target at these sizes is $\eta^\ast = 0.03$, because the refresh's failure probability grows with the occupied entries and sums to 0.007, 0.015 and 0.026 over the mission at 32, 64 and 128 agents at the runs' periods. At 32 and 64 agents the tolerance holds with more than 35 m to spare, and the run's deviation from the exact loop stays under a quarter of its envelope at every cycle. Emulated full missions at those sizes with the delay put the largest deviation's 95th percentile at 0.17 and 0.22 m, inside the radii. At 128 agents and 75 s the tolerance holds at 22 s in the norm of \cref{codesign:prop:delay} with more than 15 m to spare, the refresh being the largest part of the undelayed radius (\cref{codesign:tab:sizes,codesign:tab:delay}). The run's deviation stays under a twentieth of its envelope, and the emulated full mission with the delay puts the largest deviation's 95th percentile at 0.81 m.

\emph{The delay.} With the coordinator's feedback held for 22 s of each period, the certificate of \cref{codesign:prop:delay} holds on every graph and size of the grid where the undelayed one does, with one exception (\cref{codesign:tab:delay}). At 128 agents the delayed radius exceeds the undelayed one because the delayed norm weighs the refresh more at that size. The exception is the 4-nearest ring at weight 0.5, which passes undelayed and misses the tolerance at 22 s, almost all of it through the exact loop's held error. For that family the delay moves the coupling weight the rule can accept. In the emulated missions, which delay each command by the ring mission's measured cycle times scaled by the family's ratio from the timing runs, every family's deviation stays below its undelayed radius (\cref{codesign:tab:familiesfull}).

\subsection{Design sweep}
\label{codesign:subsec:interpretation}

The sweep of \cref{codesign:alg:codesign} varies the graph, the tolerance, the probability target and the fleet size. Its menu covers three ring dimensions, scales from 59 to 52 bits, one and two refresh iterations, the six families, the coupling weight up to the certificate's limit and tolerances from 45 to 60 m.

\emph{The graph.} At the same control on every family, the graph moves the period through the law's rotations and the admissible coupling through the held error. At a tolerance of 60 m with $\eta^\ast = 0.003$ at a coupling weight of 0.25 the families split by ring dimension. The ring, the 2-nearest ring and the torus take $N = 65\,536$ with 27 levels and a refresh every cycle at 60 s. The 4-nearest ring, the hypercube and the complete graph take $N = 131\,072$ at a scale of 58 bits with 44 levels, a refresh every fourth cycle and 75 s. With the same control, the two groups of graphs differ in every cryptographic parameter.

\emph{The tolerance.} On the ring of sixteen the returned configuration is a step function of the tolerance, with nothing returned below 47.5 m at a coupling weight of 0.1 (\cref{codesign:app:experiments}). A tighter tolerance calls for more levels, a longer interval between refreshes and a longer period, because each refresh adds error and a longer period adds contraction per cycle.

\emph{The probability target.} At the worked example the tolerance alone admits a refresh every cycle, and it is the target $\eta^\ast$ that sets the interval. Each two-iteration refresh fails with probability $2.1 \times 10^{-5}$ at $N = 131\,072$, so a mission of 524 cycles needs $T \ge 3$ at $\eta^\ast = 0.005$, which the depth bound reaches at 39 levels (\cref{codesign:fig:feasible}). The target also moves the scale. A scale of 58 bits holds 44 levels at $N = 131\,072$ where 59 bits needs $N = 262\,144$, and at $\eta^\ast = 0.003$ the rule takes it.

\emph{The fleet size.} The coordinator's cycle does not grow with the fleet on the ring, because the circuit's products and rotations act on the packed fleet at once. The time model's cycle after a refresh is 25 s at 32, 64 and 128 agents, where the field is evaluated twice and differenced, against the measured cycles of \cref{codesign:tab:sizes}. What grows is the agents' side, the encryptions and the decryptions, all linear in $M$ and outside the coordinator's period. The refresh grows with the occupied entries, 44 s at 128 and 256 entries, 46 s at 512 and 51 s at 1\,024. The period the rule returns therefore moves with the fleet size, 65 s at sixteen agents, 70 s at 32 and 64, and 75 s at 128.

\section{Conclusion}
\label{codesign:sec:conclusion}

A coordinator that keeps an encrypted estimate of a nonlinear fleet is tied to its cryptosystem in a cycle that runs from the tolerance through the refresh interval, the depth, the ring and the period back to the contraction. We stated the cycle as one bound whose constants are derived rather than measured, with the mandated flooding a term whose coefficient the controller sets. We inverted it into an algorithm that returns the configuration of least period or the constraint that forbids one. The realization ran the other way, from the cost model to the pipeline. On the ring of sixteen a mission of 524 cycles ran encrypted at the 128-bit setting inside its certified radius. The certificate holds on six coordination graphs and at 32, 64 and 128 agents, where evaluating the prediction about each agent's reference keeps the released error below the sixteen-agent level.

The next steps are to make the certificate's region a design variable of the algorithm and to hide the refresh behind a held command, so that its time no longer sets the period.

\bibliographystyle{ACM-Reference-Format}
\bibliography{references}

@misc{schlottkelakemper2024securecompute,
  title={{S}ecure numerical computations using fully homomorphic encryption},
  author={Michael Schlottke-Lakemper and Arseniy Kholod},
  year={2024},
  note={JuliaCon 2024, Eindhoven, 10th July 2024},
  howpublished={\url{https://github.com/hpsc-lab/talk-2024-juliacon-secure_numerical_computations}},
  doi={10.5281/zenodo.12703229}
}

@inproceedings{cheon2017homomorphic,
	title        = {Homomorphic encryption for arithmetic of approximate numbers},
	author       = {Cheon, Jung Hee and Kim, Andrey and Kim, Miran and Song, Yongsoo},
	booktitle    = {International conference on the theory and application of cryptology and information security},
	pages        = {409--437},
	year         = {2017},
	organization = {Springer}
}

@inproceedings{halevi2014algorithms,
	title        = {Algorithms in helib},
	author       = {Halevi, Shai and Shoup, Victor},
	booktitle    = {Annual Cryptology Conference},
	pages        = {554--571},
	year         = {2014},
	organization = {Springer}
}

@inproceedings{kogiso2015cyber,
	title        = {Cyber-security enhancement of networked control systems using homomorphic encryption},
	author       = {Kogiso, Kiminao and Fujita, Takahiro},
	booktitle    = {2015 54th IEEE Conference on Decision and Control (CDC)},
	pages        = {6836--6843},
	year         = {2015},
	organization = {IEEE}
}

@inproceedings{schluter2024code,
	title        = {A code-driven tutorial on encrypted control: From pioneering realizations to modern implementations},
	author       = {Schl{\"u}ter, Nils and Kim, Junsoo and Darup, Moritz Schulze},
	booktitle    = {2024 European Control Conference (ECC)},
	pages        = {914--920},
	year         = {2024},
	organization = {IEEE}
}

@article{kim2016encrypting,
	title     = {Encrypting controller using fully homomorphic encryption for security of cyber-physical systems},
	author    = {Kim, Junsoo and Lee, Chanhwa and Shim, Hyungbo and Cheon, Jung Hee and Kim, Andrey and Kim, Miran and Song, Yongsoo},
	journal   = {IFAC-PapersOnLine},
	volume    = {49},
	number    = {22},
	pages     = {175--180},
	year      = {2016},
	publisher = {Elsevier}
}

@inproceedings{alexandru2018cloud,
	title        = {Cloud-based MPC with encrypted data},
	author       = {Alexandru, Andreea B and Morari, Manfred and Pappas, George J},
	booktitle    = {2018 IEEE conference on decision and control (CDC)},
	pages        = {5014--5019},
	year         = {2018},
	organization = {IEEE}
}

@article{binfet2023towards,
	title     = {Towards privacy-preserving cooperative control via encrypted distributed optimization},
	author    = {Binfet, Philipp and Adamek, Janis and Schl{\"u}ter, Nils and Schulze Darup, Moritz},
	journal   = {at-Automatisierungstechnik},
	volume    = {71},
	number    = {9},
	pages     = {736--747},
	year      = {2023},
	publisher = {De Gruyter}
}

@article{farokhi2016secure,
	title     = {Secure and private cloud-based control using semi-homomorphic encryption},
	author    = {Farokhi, Farhad and Shames, Iman and Batterham, Nathan},
	journal   = {IFAC-PapersOnLine},
	volume    = {49},
	number    = {22},
	pages     = {163--168},
	year      = {2016},
	publisher = {Elsevier}
}

@article{schlor2024bootstrapping,
	title     = {Bootstrapping guarantees: Stability and performance analysis for dynamic encrypted control},
	author    = {Schlor, Sebastian and Allg{\"o}wer, Frank},
	journal   = {IEEE Control Systems Letters},
	volume    = {8},
	pages     = {2235--2240},
	year      = {2024},
	publisher = {IEEE}
}

@article{damera2026end,
	title   = {An End-to-End Encrypted Control Pipeline for Multi-Agent Coordination via CKKS Homomorphic Encryption},
	author  = {Damera, Sai Sandeep and Charitidou, Maria and Zoulkarni, Asim and Baras, John S},
	journal = {arXiv preprint arXiv:2606.07375},
	year    = {2026}
}

@inproceedings{albadawi2022openfhe,
  title        = {{OpenFHE}: Open-Source Fully Homomorphic Encryption Library},
  author       = {Al Badawi, Ahmad and Bates, Jack and Bergamaschi, Flavio and Cousins, David Bruce and Erabelli, Saroja and Genise, Nicholas and Halevi, Shai and Hunt, Hamish and Kim, Andrey and Lee, Yongwoo and Liu, Zeyu and Micciancio, Daniele and Quah, Ian and Polyakov, Yuriy and R. V., Saraswathy and Rohloff, Kurt and Saylor, Jonathan and Suponitsky, Dmitriy and Triplett, Matthew and Vaikuntanathan, Vinod and Zucca, Vincent},
  booktitle    = {Proceedings of the 10th Workshop on Encrypted Computing \& Applied Homomorphic Cryptography (WAHC '22)},
  pages        = {53--63},
  year         = {2022},
  publisher    = {ACM},
  doi          = {10.1145/3560827.3563379},
  note         = {Cryptology ePrint Archive, Paper 2022/915}
}

@article{kim2023dynamic,
  author  = {Kim, Junsoo and Shim, Hyungbo and Han, Kyoohyung},
  title   = {Dynamic Controller That Operates Over Homomorphically Encrypted
             Data for Infinite Time Horizon},
  journal = {IEEE Transactions on Automatic Control},
  volume  = {68},
  number  = {2},
  pages   = {660--672},
  year    = {2023},
  doi     = {10.1109/TAC.2022.3142124},
  eprint  = {1912.07362},
  archivePrefix = {arXiv},
  primaryClass  = {eess.SY}
}

@inproceedings{cheon2018bootstrapping,
  author    = {Cheon, Jung Hee and Han, Kyoohyung and Kim, Andrey and
               Kim, Miran and Song, Yongsoo},
  title     = {Bootstrapping for Approximate Homomorphic Encryption},
  booktitle = {Advances in Cryptology -- EUROCRYPT 2018},
  series    = {Lecture Notes in Computer Science},
  volume    = {10820},
  pages     = {360--384},
  publisher = {Springer},
  year      = {2018},
  doi       = {10.1007/978-3-319-78381-9_14}
}

@inproceedings{lee2022highprecision,
  author    = {Lee, Yongwoo and Lee, Joon-Woo and Kim, Young-Sik and
               Kim, Yongjune and No, Jong-Seon and Kang, HyungChul},
  title     = {High-Precision Bootstrapping for Approximate Homomorphic
               Encryption by Error Variance Minimization},
  booktitle = {Advances in Cryptology -- EUROCRYPT 2022},
  series    = {Lecture Notes in Computer Science},
  volume    = {13275},
  pages     = {551--580},
  publisher = {Springer},
  year      = {2022},
  doi       = {10.1007/978-3-031-06944-4_19},
  eprint    = {2020/1549},
  note      = {IACR ePrint 2020/1549}
}

@inproceedings{adamek2024encrypted,
  author    = {Adamek, Janis and Binfet, Philipp and Schl{\"u}ter, Nils and Schulze Darup, Moritz},
  title     = {Encrypted System Identification as-a-Service via Reliable Encrypted Matrix Inversion},
  booktitle = {Proceedings of the 63rd IEEE Conference on Decision and Control (CDC)},
  year      = {2024},
  eprint    = {2410.20575},
  archivePrefix = {arXiv},
  primaryClass  = {eess.SY},
  note      = {arXiv:2410.20575},
}

@inproceedings{alexandru2019encrypted,
  author    = {Alexandru, Andreea B. and Pappas, George J.},
  title     = {Encrypted {LQG} using Labeled Homomorphic Encryption},
  booktitle = {Proceedings of the 10th ACM/IEEE International Conference on Cyber-Physical Systems (ICCPS)},
  pages     = {129--140},
  year      = {2019},
  publisher = {ACM},
  doi       = {10.1145/3302509.3311049},
}

@incollection{alexandru2020secure,
  author    = {Alexandru, Andreea B. and Pappas, George J.},
  title     = {Secure Multi-Party Computation for Cloud-Based Control},
  booktitle = {Privacy in Dynamical Systems},
  editor    = {Farokhi, Farhad},
  publisher = {Springer},
  address   = {Singapore},
  pages     = {179--207},
  year      = {2020},
  doi       = {10.1007/978-981-15-0493-8_9},
}

@article{schulzedarup2021encrypted,
  author  = {Schulze Darup, Moritz and Alexandru, Andreea B. and Quevedo, Daniel E. and Pappas, George J.},
  title   = {Encrypted Control for Networked Systems: An Illustrative Introduction and Current Challenges},
  journal = {IEEE Control Systems Magazine},
  volume  = {41},
  number  = {3},
  pages   = {58--78},
  year    = {2021},
  doi     = {10.1109/MCS.2021.3076358},
}

@inproceedings{costache2023precision,
  author    = {Costache, Anamaria and Curtis, Benjamin R. and Hales, Erin and Murphy, Sean and Ogilvie, Tabitha and Player, Rachel},
  title     = {On the Precision Loss in Approximate Homomorphic Encryption},
  booktitle = {Selected Areas in Cryptography (SAC 2023)},
  series    = {Lecture Notes in Computer Science},
  volume    = {14201},
  pages     = {325--345},
  year      = {2024},
  publisher = {Springer},
  doi       = {10.1007/978-3-031-53368-6_16},
  note      = {Cryptology ePrint Archive, Paper 2022/162},
  eprint    = {2022/162},
  eprinttype = {iacr},
}

@article{murphy2024centrallimit,
  author    = {Murphy, Sean and Player, Rachel},
  title     = {A Central Limit Approach for Ring-{LWE} Noise Analysis},
  journal   = {IACR Communications in Cryptology},
  volume    = {1},
  number    = {2},
  year      = {2024},
  doi       = {10.62056/a6n59qgxq},
  note      = {Earlier version: Cryptology ePrint Archive, Paper 2019/452, ``A Central Limit Framework for Ring-LWE Decryption''},
  eprint    = {2019/452},
  eprinttype = {iacr},
}

@inproceedings{costache2020heuristic,
  author    = {Costache, Anamaria and Laine, Kim and Player, Rachel},
  title     = {Evaluating the Effectiveness of Heuristic Worst-Case Noise Analysis in {FHE}},
  booktitle = {Computer Security -- ESORICS 2020},
  series    = {Lecture Notes in Computer Science},
  volume    = {12309},
  pages     = {546--565},
  year      = {2020},
  publisher = {Springer},
  doi       = {10.1007/978-3-030-59013-0_27},
  note      = {Cryptology ePrint Archive, Paper 2019/493},
  eprint    = {2019/493},
  eprinttype = {iacr},
}

@inproceedings{bossuat2021bootstrapping,
  author    = {Bossuat, Jean-Philippe and Mouchet, Christian and Troncoso-Pastoriza, Juan and Hubaux, Jean-Pierre},
  title     = {Efficient Bootstrapping for Approximate Homomorphic Encryption with Non-Sparse Keys},
  booktitle = {Advances in Cryptology -- EUROCRYPT 2021},
  series    = {Lecture Notes in Computer Science},
  volume    = {12696},
  pages     = {587--617},
  year      = {2021},
  publisher = {Springer},
  doi       = {10.1007/978-3-030-77870-5_21},
  note      = {Cryptology ePrint Archive, Paper 2020/1203},
  eprint    = {2020/1203},
  eprinttype = {iacr},
}

@inproceedings{li2021security,
  author    = {Li, Baiyu and Micciancio, Daniele},
  title     = {On the Security of Homomorphic Encryption on Approximate Numbers},
  booktitle = {Advances in Cryptology -- EUROCRYPT 2021},
  series    = {Lecture Notes in Computer Science},
  volume    = {12696},
  pages     = {648--677},
  publisher = {Springer},
  year      = {2021},
  doi       = {10.1007/978-3-030-77870-5_23},
}

@inproceedings{li2022securing,
  author    = {Li, Baiyu and Micciancio, Daniele and Schultz, Mark and Sorrell, Jessica},
  title     = {Securing Approximate Homomorphic Encryption Using Differential Privacy},
  booktitle = {Advances in Cryptology -- CRYPTO 2022, Part I},
  series    = {Lecture Notes in Computer Science},
  volume    = {13507},
  pages     = {560--589},
  publisher = {Springer},
  year      = {2022},
  doi       = {10.1007/978-3-031-15802-5_20},
  note      = {Cryptology ePrint Archive, Paper 2022/816},
}

@inproceedings{cheon2024attacks,
  author    = {Cheon, Jung Hee and Choe, Hyeongmin and Passel{\`e}gue, Alain and Stehl{\'e}, Damien and Suvanto, Elias},
  title     = {Attacks Against the {IND-CPA}$^{\mathrm{D}}$ Security of Exact {FHE} Schemes},
  booktitle = {Proceedings of the 2024 ACM SIGSAC Conference on Computer and Communications Security (CCS)},
  pages     = {2505--2519},
  publisher = {ACM},
  year      = {2024},
  doi       = {10.1145/3658644.3690341},
  note      = {Cryptology ePrint Archive, Paper 2024/127},
}

@inproceedings{guo2024keyrecovery,
  author    = {Guo, Qian and Nabokov, Denis and Suvanto, Elias and Johansson, Thomas},
  title     = {Key Recovery Attacks on Approximate Homomorphic Encryption with Non-Worst-Case Noise Flooding Countermeasures},
  booktitle = {33rd USENIX Security Symposium (USENIX Security 24)},
  pages     = {7447--7461},
  publisher = {USENIX Association},
  year      = {2024},
}

@incollection{alfriend2010-nonlinear-spacecraft,
title = {Chapter 4 - Nonlinear Models of Relative Dynamics},
booktitle = {Spacecraft Formation Flying},
publisher = {Butterworth-Heinemann},
address = {Oxford},
pages = {59-82},
year = {2010},
isbn = {978-0-7506-8533-7},
doi = {https://doi.org/10.1016/B978-0-7506-8533-7.00209-8},
url = {https://www.sciencedirect.com/science/article/pii/B9780750685337002098},
author = {Kyle T. Alfriend and Srinivas R. Vadali and Pini Gurfil and Jonathan P. How and Louis S. Breger}
}

@article{teranishi2022designing,
	title     = {Designing optimal key lengths and control laws for encrypted control systems based on sample identifying complexity and deciphering time},
	author    = {Teranishi, Kaoru and Sadamoto, Tomonori and Chakrabortty, Aranya and Kogiso, Kiminao},
	journal   = {IEEE Transactions on Automatic Control},
	volume    = {68},
	number    = {4},
	pages     = {2183--2198},
	year      = {2022},
	publisher = {IEEE}
}

@article{murguia2020secure,
	title     = {Secure and private implementation of dynamic controllers using semihomomorphic encryption},
	author    = {Murguia, Carlos and Farokhi, Farhad and Shames, Iman},
	journal   = {IEEE Transactions on Automatic Control},
	volume    = {65},
	number    = {9},
	pages     = {3950--3957},
	year      = {2020},
	publisher = {IEEE}
}

@chapter{HE-standard,
	publisher = {Springer},
	series    = {Springer Books},
	edition   = {None},
	booktitle = {Protecting Privacy through Homomorphic Encryption},
	chapter   = {None},
	author    = {Martin Albrecht and Melissa Chase and Hao Chen and Jintai Ding and Shafi Goldwasser and Sergey Gorbunov and Shai Halevi and Jeffrey Hoffstein and Kim Laine and Kristin Lauter and Satya Lokam and Daniele Micciancio and Dustin},
	title     = {Homomorphic Encryption Standard},
	year      = {2021},
	month     = {January},
	pages     = {31-62},
	volume    = {None},
	doi       = {10.1007/978-3-030-77287-1_2},
	url       = {https://ideas.repec.org/h/spr/sprchp/978-3-030-77287-1_2.html}
}

\clearpage
\appendix
\crefalias{section}{appendix}
\section{The Refresh Error}
\label{codesign:app:refresh}

Consider OpenFHE 1.4.2~\cite{albadawi2022openfhe} with the flexible rescaling technique, a uniform ternary secret, a 60-bit first modulus and a fleet vector in at most $N/4$ entries, and let $n$ be the number of occupied entries and $R$ a bound on their size. One refresh changes every entry by at most $\beta(N, \Delta, n, R)$ and fails with probability at most $\etaref$ per call, outside a key event of probability at most $\etakey$. Both are computed from the library's moduli, its correction factor, its interval constant $K_{\mathrm{I}} = 512$ and its degree-88 modular-reduction polynomial. Three assumptions enter, the models of the scheme's randomness behind the constants of \cref{codesign:ass:errors}: the random component of the refreshed ciphertext, modulo the base modulus, is uniform and independent of the key; the rescale roundings are independent and uniform per coefficient; and a ternary polynomial's values over the roots of one decoded entry behave as independent complex Gaussians, which the fiber energies' tail below uses. No sample enters the bound. It has three parts: a phase error on the integer part of the modular reduction, set by the distance of the library's first prime from a power of two; a rounding term from the rescalings of the refresh; and the error of the sine approximation, a coherent shift for a near-constant vector that falls as the square of the correction factor while the other two grow with it. The random parts share one root and the deterministic parts add.

\paragraph{Constants of one refresh} Let $p = \log_2 \Delta$, $q_0$ the first modulus, $\Delta_0$ the last prime of the chain (the scale after the first rescale of the refresh), $n$ the occupied entries, $K_{\mathrm{I}} = 512$ the interval constant and $r = 6$ the double-angle steps. The library sets
\begin{align*}
\mathrm{dg} &= \operatorname{round}\bigl(\log_2 (q_0 / 2^{p})\bigr), \\
\mathrm{cf} &= \min\Bigl(13, \max\bigl(7, \operatorname{round}\bigl(19.1 - 0.265\,(2 \log_2 N + \log_2 n)\bigr)\bigr)\Bigr), \\
c &= \frac{2^{-\mathrm{dg}}}{K_{\mathrm{I}} N}, \qquad
\kappa = \frac{\lfloor c\,\Delta_0 + \tfrac{1}{2} \rfloor}{c\,\Delta_0}, \qquad
\varrho = \frac{q_0^2\,\kappa}{\Delta_0\,2^{60 + \mathrm{dg}}}, \\
G &= \frac{2^{\mathrm{cf}}\,2^{60}}{q_0}, \qquad
c_I = G\,\abs{\varrho - 1} .
\end{align*}
The integer part $I_j$ of coefficient $j$ after the modulus raise enters the sine with the coefficient $\varrho$, so each coefficient receives $G(\varrho - 1) I_j$: this is the phase error. With $\mathsf{h}$ the Hamming weight of the key and
\[
T(\mathsf{h}, t) = \min_{\lambda > 0} \exp\Bigl( \mathsf{h} \ln \frac{\sinh(\lambda/2)}{\lambda/2} - \lambda t \Bigr),
\]
$\Pr[\abs{I_j} \ge t + 1 \mid \mathsf{h}] \le 2\,T(\mathsf{h}, t)$ under the uniformity assumption. The refresh keeps the $2n$ subring coefficients, so $I_{\max} = \min\{ i \le K_{\mathrm{I}} - 1 : 4n\,T(\bar{h}, i) \le 10^{-10} \}$, or $K_{\mathrm{I}} - 1$ if no such $i$ exists.

\begin{definition}[Key event]
\label{codesign:def:keyevent}
For the secret key $s$, drawn once, let $\mathsf{h}(s)$ be its Hamming weight and, for $m$ entries, let
\[
A_m(s) = \max_j\, g_m^{-1} \sum_{\zeta \in F_j} \abs{s(\zeta)}^2
\]
be its largest fiber energy, where $F_j$ is the set of the $g_m = N/(2m)$ roots of $X^N + 1$ that the decoder averages into entry $j$. The key's failure event is
\[
\{\mathsf{h}(s) \ge \bar{h}\} \cup \{A_n(s) > \bar{A}_n\} \cup \{A_{2n}(s) > \bar{A}_{2n}\},
\]
where $\bar{h}$ and $\bar{A}_m$ solve
\[
e^{-N\,\mathrm{KL}(\bar{h}/N \,\|\, 2/3)} = \frac{\etakey}{3}, \qquad
m\,e^{-g_m (t_m - \ln(1 + t_m))} = \frac{\etakey}{3},
\]
with $\bar{A}_m = \frac{2N}{3}\,(1 + t_m)$.
Its probability is at most $\etakey$.
\end{definition}

\paragraph{The bound, one iteration} For entries bounded by $R$, let $u(m, A) = m\,(1 + A) / (12\,\Delta_{\min}^2)$ be the variance that one rescale adds to the real part of an entry, and let $z = \sqrt{2 \ln(2n / p_{\mathrm{tail}})}$ with $p_{\mathrm{tail}} = 10^{-9}$. Then
\begin{align}
\beta_1(R) &= \beta_{\sin} + \beta_{\mathrm{ser}} + \beta_{\mathrm{int}} + \beta_{\mathrm{cst}} \nonumber \\
 &\quad + z \sqrt{V_I + V_{\mathrm{adj}} + V_{\mathrm{pre}} + V_{\mathrm{CtS}} + V_{\mathrm{mod}} + V_{\mathrm{StC}}} ,
\label{codesign:eq:beta1} \\
\beta_{\sin} &= G\,\frac{(2\pi)^2}{6}\,(a_1 + a_2)^3, \qquad a_1 = q_0\,\kappa\,2^{-\mathrm{cf} - 60} R, \nonumber \\
a_2 &= \abs{\varrho - 1}\, I_{\max}\, (2n)^{1/3}, \qquad \beta_{\mathrm{ser}} = 2n\,G\,\varepsilon_6, \nonumber \\
\beta_{\mathrm{int}} &= c_I\, n\,(1 + 2 a_1), \qquad
\beta_{\mathrm{cst}} = \frac{2n\,G}{\Delta_{\min}} \sum_k g_k\,\mu_k, \nonumber \\
V_I &= c_I^2\,\frac{n}{12}\,\bar{A}_n, \qquad
V_{\mathrm{adj}} = 2 \cdot 2^{2(\mathrm{cf} - \mathrm{dg})}\, u(n, \bar{A}_n), \nonumber \\
V_{\mathrm{pre}} &= (2n)^2 (G\,g_{\mathrm{in}})^2 \cdot 2\,u(n, \bar{A}_n), \nonumber \\
V_{\mathrm{CtS}} &= \sum_{\ell} (1 + \kappa_{\mathrm{ks}}\, r^{\mathrm{C}}_{\ell})\,(2n)^2 (G\,g_{\mathrm{in}})^2 \cdot 2\,u(2n, \bar{A}_{2n}), \nonumber \\
V_{\mathrm{mod}} &= 2n\,G^2\,\bigl(S_g + 2^{-2\mathrm{dg}}\bigr)\, u(2n, \bar{A}_{2n}), \nonumber \\
V_{\mathrm{StC}} &= \sum_{\ell} (1 + \kappa_{\mathrm{ks}}\, r^{\mathrm{S}}_{\ell})\, 2^{2(\mathrm{cf} - \mathrm{dg})}\, 2n\, u(2n, \bar{A}_{2n}) + u(n, \bar{A}_n) . \nonumber
\end{align}
The sine approximation is $\beta_{\sin} + \beta_{\mathrm{ser}}$, the phase error $\beta_{\mathrm{int}}$ with $V_I$ inside the root, and the rounding $\beta_{\mathrm{cst}}$ with the five rescale variances.
Here $\varepsilon_6$ is the error of the degree-88 series after the double-angle steps, $\varepsilon_j = (4 a'_{j-1} + 2 \varepsilon_{j-1})\,\varepsilon_{j-1} + 2^{-50}$ with $a'_j = (2\pi)^{-2^{j - r}}$, starting from the series' error $\varepsilon_0$ against $(2\pi)^{-1/64}\cos\bigl((2\pi K_{\mathrm{I}} x - \pi/2)/64\bigr)$ on $[-1, 1]$. $g_{\mathrm{in}}$ is the largest slope of the modular reduction and $S_g$ the largest sum, over the operating points $x = (I + u)/K_{\mathrm{I}}$ with $\abs{I} \le I_{\max}$ and $\abs{u} \le 2^{-\mathrm{cf}} R + \abs{\varrho - 1} I_{\max}$, of the squared gains of its rescale events; $g_k \mu_k$ is the gain and the size of the rounding of its $k$-th encoded constant. $r^{\mathrm{C}}_\ell$ and $r^{\mathrm{S}}_\ell$ are the key switches of level $\ell$ of the two linear transforms, and $\kappa_{\mathrm{ks}} = 9\, d_{\mathrm{num}}\, \sigma^2 \bigl((a + 1) D / P\bigr)^2$ is the key switch's digit term in rescale units ($d_{\mathrm{num}}$ digits of $a$ towers, $D$ the largest digit, $P$ the special modulus). Outside the key event, every entry's error is at most $\beta_1(R)$ except with probability
\begin{equation}
\etaref^{(1)} = 4n\,T(\bar{h}, I_{\max}) + p_{\mathrm{tail}} \quad \text{per call.}
\label{codesign:eq:etaref1}
\end{equation}

\paragraph{Two iterations} At precision $p'$, with $2^{p'} \beta_1(R) < 1$, the library returns $m - e_2/2^{p'}$, where $e_2$ is the error of a refresh of $2^{p'} e_1$. Then
\begin{equation}
\begin{aligned}
\beta_2(R) &= \frac{\beta_1\bigl(2^{p'} \beta_1(R)\bigr)}{2^{p'}} + \frac{R\,2^{p' - 1}}{\Delta_{\min}} \\
 &\quad + \bigl(1 + 2^{-p'}\bigr)\, z\, \frac{\sqrt{(n/12)(1 + \bar{A}_n)}}{\Delta_{\min}}, \qquad
\etaref^{(2)} = 2\,\etaref^{(1)} .
\end{aligned}
\label{codesign:eq:beta2}
\end{equation}
The first term carries the second refresh's sine shift at the input range $2^{p'}\beta_1(R)$, the second the rounding of the constant $2^{-p'}$, and the third the two rescales of the result. The bound needs the modulus raise of the second refresh to stay inside $q_0$ and the series' points inside its domain, which hold at the library's precision $p' = 12$ in every configuration below. \Cref{codesign:tab:refreshconst} gives every constant at the mission's configuration and \cref{codesign:tab:refreshconfigs} the bound across the configurations this paper runs or sweeps.

\begin{table*}[t]
\caption{The refresh constants at the mission's configuration.}
\label{codesign:tab:refreshconst}
\centering
\small
\begin{tabular}{llll}
\toprule
Constant & Value & Constant & Value \\
\midrule
$q_0 - 2^{60}$ & $-262\,143$ & $\Delta_0 - 2^{59}$ & $12\,058\,625$ \\
$\mathrm{cf}$, $\mathrm{dg}$ & 8, 1 & $\kappa - 1$ & $-2.09 \times 10^{-11}$ \\
$\varrho - 1$ & $-4.23 \times 10^{-11}$ & $G$ & 256 \\
$c_I$ & $1.08 \times 10^{-8}$ & $\Delta_{\min} / 2^{59} - 1$ & $-2.95 \times 10^{-10}$ \\
$\bar{h}$ & $88\,507$ ($0.675\,N$) & $\bar{A}_n$, $\bar{A}_{2n}$ & $1.36$, $1.54$ times $2N/3$ \\
$I_{\max}$ & 511 & $z$ & 7.25 \\
$g_{\mathrm{in}}$ & 512.0 & $S_g$ & $2.47 \times 10^{10}$ \\
$\sum_k g_k \mu_k$ & $4.17 \times 10^{4}$ & $\varepsilon_0$, $\varepsilon_6$ & $1.88 \times 10^{-14}$, $1.28 \times 10^{-11}$ \\
$\kappa_{\mathrm{ks}}$ & $9.34 \times 10^{-4}$ & $r^{\mathrm{C}}_\ell$, $r^{\mathrm{S}}_\ell$ & $(10, 10, 4)$, $(10, 10, 4)$ \\
\midrule
$\beta_{\sin}$, $\beta_{\mathrm{ser}}$ & $1.00 \times 10^{-4}$, $8.41 \times 10^{-7}$ & $\beta_{\mathrm{int}}$, $\sqrt{V_I}$ & $1.40 \times 10^{-6}$, $1.22 \times 10^{-5}$ \\
$\beta_{\mathrm{cst}}$ & $4.74 \times 10^{-9}$ & $\sqrt{V_{\mathrm{adj}} + \dots + V_{\mathrm{StC}}}$ & $1.91 \times 10^{-6}$ \\
$\beta_1(1)$ & $1.92 \times 10^{-4}$ & $\etaref^{(1)}$ & $1.05 \times 10^{-5}$ \\
$p'$, $2^{p'}\beta_1(1)$ & 12, 0.787 & $\beta_1(0.787)/2^{12}$ & $3.43 \times 10^{-8}$ \\
$\beta_2(1)$ & $3.43 \times 10^{-8}$ & $\etaref^{(2)}$ & $2.1 \times 10^{-5}$ \\
$\beta_1(0.4)$, $\beta_2(0.4)$ & $9.8 \times 10^{-5}$, $2.39 \times 10^{-8}$ & $\lceil K/T \rceil\, \etaref^{(2)}$ & $3.68 \times 10^{-3}$ \\
\bottomrule
\end{tabular}
\end{table*}

\begin{table*}[t]
\caption{The refresh bound across configurations ($n = 128$, scale $2^{59}$ unless stated; $p'$ the second iteration's precision).}
\label{codesign:tab:refreshconfigs}
\centering
\small
\setlength{\tabcolsep}{4pt}
\begin{tabular}{lrrrrrr}
\toprule
Configuration & $\mathrm{cf}$ & $\beta_1$ & $\etaref^{(1)}$ & $p'$ & $\beta_2$ & $\etaref^{(2)}$ \\
\midrule
$N = 131\,072$, depth 39 (mission) & 8 & $1.92 \times 10^{-4}$ & $1.05 \times 10^{-5}$ & 12 & $3.43 \times 10^{-8}$ & $2.1 \times 10^{-5}$ \\
$N = 65\,536$, depth 27, budget $[2,2]$, 28 digits & 9 & $8.65 \times 10^{-4}$ & $1.1 \times 10^{-9}$ & 12 & $4.77 \times 10^{-7}$ & $2.2 \times 10^{-9}$ \\
$N = 65\,536$, depth 26, budget $[2,2]$, 14 digits & 9 & $9.15 \times 10^{-4}$ & $1.1 \times 10^{-9}$ & & & \\
$N = 262\,144$, depth 39 & 8 & $3.67 \times 10^{-4}$ & $7.1 \times 10^{-2}$ & & & \\
$N = 131\,072$, depth 30, $n = 16$ & 9 & $8.0 \times 10^{-5}$ & $1.31 \times 10^{-6}$ & 12 & $1.36 \times 10^{-8}$ & $2.63 \times 10^{-6}$ \\
$N = 131\,072$, depth 30, $n = 64$ & 9 & $1.46 \times 10^{-4}$ & $5.25 \times 10^{-6}$ & 11 & $5.93 \times 10^{-8}$ & $1.05 \times 10^{-5}$ \\
$N = 65\,536$, depth 30, $n = 16$, scale $2^{50}$ & 10 & $4.47 \times 10^{-3}$ & $1.1 \times 10^{-9}$ & 6 & $6.97 \times 10^{-5}$ & $2.2 \times 10^{-9}$ \\
\bottomrule
\end{tabular}
\end{table*}

At the mission the bound is set by the sine term at $\mathrm{cf} = 8$ and then by the phase deviation; in two iterations the second refresh's sine shift is divided by $2^{12}$. At $N = 65\,536$ with one tower per digit and one special prime, $D/P$ is about one and the key switch's rescale deviation sets the bound. At $N = 262\,144$ the failure bound per refresh, $7.1 \times 10^{-2}$ from the integer part over a key of Hamming weight up to $176\,355$, excludes that ring from any guarantee with $\eta < 1$ over the mission. The retained formation errors of the mission stay below 0.38 in their units of km and m/s, and over the $\lceil K/T \rceil = 175$ refreshes the bound allows (the mission made 173) the failure probability sums to $3.68 \times 10^{-3}$, which enters \cref{codesign:eq:eta} as $\lceil K/T \rceil\,\etaref$.

The bound was checked against samples at the 128-bit setting: 36 refreshes in each of three configurations of \cref{codesign:tab:refreshconfigs} (sixteen entries at $N = 131\,072$, sixteen at $N = 65\,536$ with a 50-bit scale, 64 at $N = 131\,072$), with the library printing each configuration. The criteria, fixed before the samples, were every sampled entry error within $\beta$, the root-mean-square error within the interval the phase term predicts, the constant-vector shift within four standard errors and 0.1 percent of its prediction, and the two-iteration errors within their bound. All 18 checks passed, the largest entry error between 0.27 and 0.56 of the bound, the root-mean-square error within the interval the phase term predicts, the shift of a constant vector within two standard errors of its prediction ($-2.46 \times 10^{-5}$ against $-2.51 \times 10^{-5}$ at sixteen entries), and the two-iteration errors within their bound. A decryption of the retained ciphertext before and after a refresh does not measure the refresh error on its own, because the decoder adds a Gaussian of its own at every decryption that exceeds the two-iteration bound; the comparison on the pipeline's own retained ciphertext therefore averages 64 decryptions before and after each refresh. At the mission's configuration, with entries within 0.35, the two-iteration refresh changed the entries by at most $1.4 \times 10^{-8}$ and $1.9 \times 10^{-8}$ against the bound of $2.39 \times 10^{-8}$ for entries within 0.4, and a one-iteration refresh by $2.6 \times 10^{-5}$ against $9.8 \times 10^{-5}$. The same refreshes read through single decryptions gave $1.2 \times 10^{-7}$ and $1.6 \times 10^{-7}$, the decoder's noise, which also accounts for the mission's logged values ($2.7 \times 10^{-8}$ median, $6.2 \times 10^{-8}$ largest over its 173 refreshes).

\section{The Per-Cycle Error}
\label{codesign:app:percycle}

The constants $\abar$, $\etacyc$ and $\sigfl$ are derived from the library's circuits at the configuration of \cref{codesign:tab:config} and composed along the traced record of one cycle, with the operands bounded on the region. Three facts about the decoder fix the per-entry constants.

\begin{lemma}[Decoded entry]
\label{codesign:lem:decoded}
Let a ciphertext of scale $\Delta$ with $B$ occupied entries decrypt to $\Delta\,\mu + e$, where $\mu$ is the encoded message and the coefficients of $e \in \mathbb{Z}[X]/(X^N + 1)$ are uncorrelated with mean zero and variance $V$. Let $G = N/(2B)$ and $\omega = e^{2\pi i/(4B)}$. The decoder returns as entry $j$ the real part of $\mu(\omega^{e_j}) + \Delta^{-1} \sum_{k=0}^{2B-1} e_{kG}\,\omega^{k e_j}$, with $e_j = 5^j \bmod 4B$ odd. Then (1) the error of entry $j$ has variance $B\,V/\Delta^2$; (2) it equals the mean of the $G$ copies of entry $j$ in the canonical embedding, $G^{-1}\sum_{t \equiv j \,(\mathrm{mod}\, B)} \operatorname{Re} e(\zeta^{5^t})/\Delta$ with $\zeta = e^{i\pi/N}$, each copy having variance $G\,B\,V/\Delta^2$; (3) the library's noise estimate of the ciphertext is $\log_2$ of a statistic with mean square $V/2\,(1 + O(1/B))$, so the root mean square over the $B$ entries of the imaginary error is $\sqrt{2B}\,2^{\hat{\varepsilon}}/\Delta$.
\end{lemma}
\begin{proof}
(1) The variance is
\[
\frac{V}{\Delta^2}\sum_{k=0}^{2B-1}\cos^2\frac{2\pi k e_j}{4B} = \frac{V}{\Delta^2}\Bigl(B + \tfrac{1}{2}\operatorname{Re}\sum_{k=0}^{2B-1}\omega^{2 k e_j}\Bigr) = \frac{BV}{\Delta^2},
\]
because $\omega^{2 e_j}$ is a $2B$-th root of unity other than 1; the complex entry has $\mathbb{E}\abs{\cdot}^2 = 2BV/\Delta^2$, and the decoder keeps its real part. (2) The roots $\zeta^{5^t}$ with $t \equiv j \pmod B$ are $\zeta^{c(1 + 4Bl)}$, $c = 5^j$, $l = 0, \dots, G - 1$, and
\[
G^{-1}\sum_l \zeta^{n c (1 + 4Bl)} = \zeta^{nc}\,G^{-1}\sum_l e^{2\pi i\, n c l/G}
\]
vanishes unless $G \mid n$; so the copies' mean keeps exactly the coefficients $e_{kG}$, and $\zeta^{kGc} = \omega^{kc}$. Each copy has variance $(N/2)V/\Delta^2 = GBV/\Delta^2$. (3) The estimate is one half of the sample deviation of $a_k + a_{2B-k}$, $k = 1, \dots, 2B - 1$, over the $2B$ coefficients the decoder reads, each of variance $2V$.
\end{proof}

Part (1) is the factor $B$ (the real part), not $2B$, in every per-entry constant; part (2) says that each decoded entry is one random quantity, so the failure events below run over the decoded entries with no factor for the copies; part (3) is the one half in the library's noise-estimate floor. With $\sigma = 3.19$, $N = 131\,072$ and $\Delta_{\min} = 2^{59}(1 - 2.95 \times 10^{-10})$ the per-operation constants are
\begin{align*}
V_{\mathrm{pk}} &= \sigma^2 \bigl(1 + \tfrac{4N}{3}\bigr), \quad \sqrt{B\,V_{\mathrm{pk}}}/\Delta_{\min} = 2.62 \times 10^{-14} \text{ per entry}, \\
 &\quad \text{estimate floor } \log_2\sqrt{V_{\mathrm{pk}}/2} = 9.881 \ \text{(fresh encryption)}, \\
V_{\mathrm{rs}} &= \tfrac{1}{12}\bigl(1 + \tfrac{2N}{3}\bigr), \quad 1.67 \times 10^{-15} \text{ per entry (rescale)}, \\
V_{\mathrm{ks}} &= V_{\mathrm{rs}}\,(1 + \kappa_{\mathrm{ks}}), \quad 1.68 \times 10^{-15} \text{ per entry (key switch)}, \\
V_{\mathrm{enc}} &= \tfrac{1}{12} \ \text{(encoding)} .
\end{align*}
The smallest estimates the library reported for a released command were $2^{9.87}$ in the mission and $2^{9.98}$ in the probe, against the floor $2^{9.88}$.

\paragraph{The recursion along the trace} The bound runs the pipeline's own code on a bounding backend. Every ciphertext $x$ holds, per entry, a magnitude bound $m_x$ on the region and a deviation bound $s_x$ that holds under any dependence between the errors it combines. With the shorthand $u_V = B V / \Delta_{\min}^2$, the rules are the following. An upload of the $M$ agents' blocks has $s^2 = M\,u_{V_{\mathrm{pk}} + V_{\mathrm{enc}}}$. A rescale of an operand adds $u_{V_{\mathrm{rs}}}$ to $s^2$. An addition or subtraction has $s = s_a + s_b$ and $m = m_a + m_b$. A product by a plaintext $c$ has $s^2 = (\abs{c}\, s_a)^2 + m_a^2\, u_{V_{\mathrm{enc}}}$. A product of ciphertexts has $s^2 = (m_a s_b + m_b s_a + c_{\times}\, s_a s_b)^2 + B\,V_{\mathrm{ks}}/\Delta_{\min}^4$. A rotation by $d$ has $s^2 \leftarrow \operatorname{circshift}(s, d)^2 + B\,V_{\mathrm{ks}}/\Delta_{\min}^{2\deg}$. The products of two errors are sub-exponential rather than sub-Gaussian, and their share is bounded separately. The second-order term of a product is a mean of products of per-copy errors; at $c_\times = \sqrt{3}\,G = 887$, which bounds it for strictly sub-Gaussian sources, no per-entry bound changes by more than $3.5 \times 10^{-6}$ relative, so $c_\times = 1$ is used. A difference with an uploaded reference is a formation error, bounded from the region; the retained formation error is bounded by 1 km and 1 m/s per entry, the states by the region of \cref{codesign:subsec:agents}, and the uploaded applied command by $10^{-2}$ m/s$^2$ per entry, which the flooding keeps (below). The cycle starts from an exact retained estimate at the two-iteration refresh's output level. On the ring of sixteen the record has 123 additions, 217 encodings, 65 encryptions, 52 ciphertext products, 16 plaintext products, 38 rotations and 4 subtractions per cycle.

\begin{lemma}[Bias and fresh part]
\label{codesign:lem:biasfresh}
Split the error sources of one cycle into $\mathcal{S}_{\mathrm{c}}$, the encryption and encoding errors of the $n_{\mathrm{c}}/B$ coefficient ciphertexts uploaded once before the first cycle, and $\mathcal{S}_{\mathrm{f}}$, all other sources (the cycle's uploads, the release's encryption of zero, rescales, key switches, plaintext encodings). Let $s^{\mathrm{c}}$ and $s^{\mathrm{f}}$ be the deviation bounds of the recursion with only $\mathcal{S}_{\mathrm{c}}$, respectively only $\mathcal{S}_{\mathrm{f}}$, kept, and let $z_{\mathrm{c}} = \sqrt{2 \ln (2 n_{\mathrm{c}} / \eta_{\mathrm{coef}})}$ and $z = \sqrt{2 \ln (2 n_{\mathrm{e}} / \etacyc)}$. Outside an event of probability $\eta_{\mathrm{coef}}$, fixed for the mission, and for every cycle outside an event of probability $\etacyc$, the error of the cycle from any joint state in $\errbox \times \errbox$ satisfies, per entry $i$ of the retained estimate and of the released commands, $\abs{a_i(k)} \le z_{\mathrm{c}}\, s^{\mathrm{c}}_i + z\, s^{\mathrm{f}}_i$.
\end{lemma}
\begin{proof}
Each coefficient entry's decoded error $\epsilon_j$ is drawn once with deviation $\sigma_j$. The bound $s^{\mathrm{c}}_i$ is a sum over paths of $\abs{\text{gain}} \cdot \sigma_j$, so on $\{\abs{\epsilon_j} \le z_{\mathrm{c}}\sigma_j \ \forall j\}$ the coefficients' part of $a_i(k)$ is at most $z_{\mathrm{c}} s^{\mathrm{c}}_i$ at every cycle and every state of the region; the second-order terms are covered by the check at $c_\times = \sqrt{3}\,G$. The union over the $n_{\mathrm{c}}$ entries with the sub-Gaussian tail $2e^{-z^2/2}$ gives $\eta_{\mathrm{coef}}$. The sources in $\mathcal{S}_{\mathrm{f}}$ are drawn afresh each cycle, independent and strictly sub-Gaussian (uniform roundings, and Gaussians conditionally on the fiber energies bounded below), and $s^{\mathrm{f}}_i$ bounds the deviation of their sum, so the tail per entry is $2e^{-z^2/2}$ and the union over the $n_{\mathrm{e}}$ entries gives $\etacyc$.
\end{proof}

The variances above hold at the mean fiber energy $2N/3$. On the key event (\cref{codesign:def:keyevent}) the key's fiber energy is at most $\bar{A}_B$, and outside an event of probability $\eta_v$ every encryption's ternary $v$ over the mission has fiber energy at most $\bar{A}_v$, the bound of the definition with $\eta_v$ divided by the mission's count of encryptions; the bound then uses $V_{\mathrm{pk}} = \sigma^2 (1 + \bar{A}_v + \bar{A}_B)$, $V_{\mathrm{rs}} = \tfrac{1}{12}(1 + \bar{A}_B)$ and $V_{\mathrm{ks}} = V_{\mathrm{rs}}(1 + \kappa_{\mathrm{ks}})$. With per-agent bounds in physical units $b_E = (z_{\mathrm{c}} s^{\mathrm{c}}_E + z\, s^{\mathrm{f}}_E) \circ S$ on the retained estimate and $b_U = (z_{\mathrm{c}} s^{\mathrm{c}}_U + z\, s^{\mathrm{f}}_U)\, S_u$ on the released command, each the largest over the agents, an estimate-only error moves $(e_x, \tilde{x})$ by $(0, -b_E)$ and a command error moves plant and estimate alike by $(H_u b_U, 0)$, because the agent uploads the command it applied. Hence
\begin{equation}
\begin{aligned}
\abar = \sqrt{M}\Bigl[&\max_{\sigma \in \{\pm 1\}^6} \sqrt{\theta\,(b_E \circ \sigma)^\top P_e\,(b_E \circ \sigma)} \\
 &+ \max_{\sigma \in \{\pm 1\}^3} \sqrt{(H_u (b_U \sigma))^\top P_r\,(H_u (b_U \sigma))}\Bigr],
\end{aligned}
\label{codesign:eq:abar}
\end{equation}
where $P = I_M \otimes \operatorname{blkdiag}(P_r, \theta P_e)$ in the coordinates $(e_x, \tilde{x})$ with $\tilde{x} = e_x - e_{\hat{x}}$, $\theta$ the smallest weight on a logarithmic grid at which the linear part contracts at the larger of the two blocks' rates, $H_u = hB + \tfrac{h^2}{2} A_c B$ is Heun's input matrix with $A_c$ the linear part of the field, the Clohessy and Wiltshire matrix, and $S$, $S_u$ are the encoding scales of \cref{codesign:tab:encoding}. At the mission $n_{\mathrm{c}} = 35 \times 128$ entries give $z_{\mathrm{c}} = 7.72$ at $\eta_{\mathrm{coef}} = 10^{-9}$, and $n_{\mathrm{e}} = 2Mn + Mm = 240$ decoded entries give $z = 7.33$ at $\etacyc = 10^{-9}$ per cycle; the mission's $34\,620$ encryptions give $\bar{A}_v = 1.42 \cdot 2N/3$ at $\eta_v = 10^{-9}$. \Cref{codesign:tab:percycle} gives the bounds. Of $\abar$, 98.5 percent is the coefficient ciphertexts' error, one draw for the mission; the fresh part of each cycle is 1.5 percent. With $\abar = 3.08 \times 10^{-5}$ and the plant's own sampled map as the certified map (\cref{codesign:app:proofs}), the tube at the mission has $\rpos = 1.52 + 0.33 + 0.06 = 1.91$ m against the margin of 5.86 m; on the Heun-model map the same constant gives 1.83 m. The flooding's 1.52 m is a position standard deviation of 0.216 m from \cref{codesign:eq:cov} at the quantile 7.0 over the agent-cycles at $\etapos = 10^{-6}$. The arithmetic and the remainder are 0.33 m, of which the arithmetic constant $\abar$ is a bias of the coefficient ciphertexts bounded once for the mission, and the refresh is 0.06 m.

\begin{table*}[t]
\caption{The per-cycle error at the mission: per-entry bounds in encoded units, $b_E$ and $b_U$ in physical units, $\abar$ in $\norm{\cdot}_P$.}
\label{codesign:tab:percycle}
\centering
\footnotesize
\setlength{\tabcolsep}{4pt}
\begin{tabular}{lll}
\toprule
Quantity & Mean fiber energy & Bounded fiber energies \\
\midrule
$s^{\mathrm{c}}_E$ (position; velocity) & $(3.45, 1.85, 47.3) \times 10^{-10}$; $(1.04, 1.06, 1.34) \times 10^{-8}$ & $(4.07, 2.18, 55.7) \times 10^{-10}$; $(1.23, 1.25, 1.58) \times 10^{-8}$ \\
$s^{\mathrm{f}}_E$ (position; velocity) & $(5.74, 3.17, 75.8) \times 10^{-12}$; $(1.68, 1.70, 2.15) \times 10^{-10}$ & $(6.69, 3.69, 88.4) \times 10^{-12}$; $(1.96, 1.98, 2.51) \times 10^{-10}$ \\
$s^{\mathrm{c}}_U$, $s^{\mathrm{f}}_U$ & $3.56 \times 10^{-12}$, $8.79 \times 10^{-14}$ & $4.20 \times 10^{-12}$, $1.03 \times 10^{-13}$ \\
$b_E$ & $(2.71, 1.45, 37.1) \times 10^{-6}$ m; $(8.18, 8.29, 10.5) \times 10^{-8}$ m/s & $(3.19, 1.71, 43.7) \times 10^{-6}$ m; $(9.64, 9.77, 12.4) \times 10^{-8}$ m/s \\
$b_U$ & $5.63 \times 10^{-13}$ m/s$^2$ & $6.64 \times 10^{-13}$ m/s$^2$ \\
$\abar$ (bias, fresh) & $2.61 \times 10^{-5}$ ($2.57 \times 10^{-5}$, $3.9 \times 10^{-7}$) & $3.08 \times 10^{-5}$ ($3.04 \times 10^{-5}$, $4.6 \times 10^{-7}$) \\
$z_{\mathrm{c}}$, $z$ & 7.72, 7.33 & 7.72, 7.33 \\
\bottomrule
\end{tabular}
\end{table*}

\paragraph{The flooding} The library sizes it in the coefficient embedding as $\sigma_{\mathrm{fl,coef}} = c_{\mathrm{fl}}\,2^{\hat{\varepsilon}}$ with $c_{\mathrm{fl}} = \sqrt{12 n_{\mathrm{d}}}\,2^{s/2}$~\cite{albadawi2022openfhe}, following~\cite{li2022securing}, where $2^{\hat{\varepsilon}}$ is the library's estimate of the released command's error in the ciphertext's raw units at scale $\Delta$, $s$ the statistical security parameter and $n_{\mathrm{d}}$ the number of decryptions the key serves, a deployment parameter taken as given at the value the mission used. We take $s = 30$ and $n_{\mathrm{d}} = 1$, so $c_{\mathrm{fl}} = 2^{16.79}$, and $c_{\mathrm{fl}}$ scales as $\sqrt{n_{\mathrm{d}}}$. The estimate is obtained by the library's own estimation phase, the same computation run under a calibration key with the errors of its decryptions measured, which the fleet's operator runs before the mission on the agents' side; in the mission run it is the largest error of a released command over thirty cycles, $2^{25.87}$. Per released entry, after decoding,
\begin{equation}
\sigfl = \sqrt{2B}\; \frac{\sigma_{\mathrm{fl,coef}}}{\Delta}\; S_u ,
\label{codesign:eq:sigfl}
\end{equation}
where $\sqrt{B}$ is the real part of the decoded entry (\cref{codesign:lem:decoded}), $\sqrt{2}$ is the decoder's own Gaussian, which it adds at every decryption at the same size as the flooding, and $S_u = 0.02$ m/s$^2$ per encoded unit is the command's encoding scale: $3.86 \times 10^{-6}$ m/s$^2$ at the mission's estimate. The form is checked at the 128-bit setting against the spread of 64 decryptions of a released command, $1.04 \times 10^{-3}$ encoded per entry measured against $1.04 \times 10^{-3}$ predicted at the probe's estimate of $2^{28.3}$. The estimate is set by the prediction's arithmetic at the train's two end agents, where the state's monomials are largest, which is what the realization of \cref{codesign:subsec:estimator} that evaluates the field about the reference lowers.

\paragraph{The command box} The uploaded applied command carries the flooding, so \cref{codesign:ass:errors} holds on the region with the command box enlarged by it: outside an event of probability $\eta_\zeta$ over the $3MK$ released entries,
\[
\abs{u^{\mathrm{app}}_{\mathrm{enc}, i}(k)} \le \max_{i,k}\abs{u^{\mathrm{app}}_i(k)} + \bar{\delta}_u + \bar{\varepsilon} + z_\zeta\,\sigfl \le 10^{-2}\ \text{m/s}^2 ,
\]
with $z_\zeta = \sqrt{2 \ln (2 \cdot 3MK / \eta_\zeta)}$.
At the mission the terms are $4.76 \times 10^{-3}$ (the exact loop's largest applied command with its feedforward), $2.18 \times 10^{-4}$ (the law's response to a disturbance within $\dhat$), $5.3 \times 10^{-13}$ (the release's error) and $2.77 \times 10^{-5}$ ($z_\zeta = 7.02$ at $\eta_\zeta = 10^{-6}$), $5.01 \times 10^{-3}$ m/s$^2$ in all.

\paragraph{The probability} The terms of \cref{codesign:eq:eta} at the mission: $\etakey = 10^{-9}$ (the key's failure event); $\lceil K/T \rceil \etaref = 175 \times 2.1 \times 10^{-5} = 3.68 \times 10^{-3}$ (some refresh's integer part outside $I_{\max}$ or its random part beyond $z = 7.25$); $\eta_{\mathrm{coef}} = 10^{-9}$ (some coefficient entry's error beyond $z_{\mathrm{c}}$ deviations, once); $K\,\etacyc = 5.2 \times 10^{-7}$ (some cycle's fresh error beyond $z$ deviations in one of its 240 entries); $\eta_v = 10^{-9}$; $\eta_\zeta = 10^{-6}$; $\etag = 10^{-6}$ (the Gaussian part's $P$-norm beyond $\hat{g}$ at some cycle); $\etapos = 10^{-6}$ (some agent-cycle's position part beyond $q = 7.01$). Their sum is $\eta = 3.68 \times 10^{-3}$, almost all of it the refreshes.

\section{Proofs}
\label{codesign:app:proofs}

\paragraph{The certified map} Heun's step is affine in the command, because its stage's positions $q + h\dot{q}$ do not contain it; the flow is affine only to first order, and $L_B$ and $L_\nu$ bound the rest by Gr\"onwall's inequality on the flow's second variation. At the mission they add $2.0 \times 10^{-4}$ m to $\abar$ (in metres of the worst position).

\begin{corollary}[Linear agents]
\label{codesign:cor:linear}
For $f(x) = Ax$ and a steady-state gain, let $A_{\mathrm{cl}}$ be the closed-loop matrix of the joint estimator and controller dynamics and $\norm{\cdot}_P$ a norm in which it is a contraction. Then \cref{codesign:ass:contraction} holds globally with $\ell = 1$, $L_2 = 0$, a constant $B_k = B$, $c = 1$ and $\gamma = \norm{A_{\mathrm{cl}}}_P$. Then $\Sigma(k+1) = A_{\mathrm{cl}}\,\Sigma(k)\,A_{\mathrm{cl}}^{\top} + \sigfl^2 B B^{\top}$, $\dhat = \ghat + \abar/(1 - \gamma) + \beta'/(1 - \gamma^T)$ and $\rpos = \rg + \kappa_P\,\bigl(\abar/(1 - \gamma) + \beta'/(1 - \gamma^T)\bigr)$.
\end{corollary}

\begin{proof}[Proof of \cref{codesign:lem:additive}]
Fix a cycle $k$ at which $z_{\mathrm{enc}}(k) \in \errbox \times \errbox$. Both loops see the same noise and disturbance realizations and move the plant under the commands the coordinator returns. The exact map applies the commands without error. The encrypted loop applies the released commands with their encryption error and the flooding $\zeta(k)$, uploads the commands as applied, and evaluates the estimator and the law on ciphertexts. On the arithmetic event of \cref{codesign:ass:errors} for that cycle, the joint state it reaches differs from $G_k(z_{\mathrm{enc}}(k), \zeta(k))$ by at most $\abar$ in $\norm{\cdot}_P$. By the expansion in the command,
\begin{align*}
G_k(z_{\mathrm{enc}}(k), \zeta(k)) = {}& G_k(z_{\mathrm{enc}}(k)) + B_k(z(k))\,\zeta(k) \\
&+ \bigl(B_k(z_{\mathrm{enc}}(k)) - B_k(z(k))\bigr)\zeta(k) \\
&+ n_k(z_{\mathrm{enc}}(k), \zeta(k)),
\end{align*}
and on the command-box event every entry of $\zeta(k)$ is at most $z_\zeta\sigfl$, so the last two terms have $P$-norm at most $\abar_\nu(\norm{d(k)}_P)$. When a refresh follows the release of cycle $k$, it reaches $z(k+1)$ through the prediction and the update, by at most $\beta'$ on the refresh's event (\cref{codesign:lem:refreshpred}). Call the sum of these differences $a(k)$. By the remainder bound of \cref{codesign:ass:contraction} at $z = z(k)$ and $d = d(k)$, both in the region, $G_k(z_{\mathrm{enc}}(k)) = G_k(z(k)) + J_k\,d(k) + \rho(k)$ with $\norm{\rho(k)}_P \le L_2 \norm{d(k)}_P^2$. Subtracting $z(k+1) = G_k(z(k))$ gives \cref{codesign:eq:dev}.
\end{proof}

\begin{proof}[Proof of \cref{codesign:thm:tube}]
The series of \cref{codesign:eq:rb} run over blocks of $\ell$ cycles, the inputs of a block accumulated by $c_\ell$ and the blocks contracted by $\gamma$, and \cref{codesign:eq:rstar} extends the bound from the block ends to the cycles inside a block. The bound $\dhat$ on the whole deviation is needed once, for the second-order term. Define $g(0) = 0$ and $g(k+1) = J_k\,g(k) + B_k\,\zeta(k)$ with $J_k = J_k(z(k))$ and $B_k = B_k(z(k))$ along the exact trajectory, and $r(k) = d(k) - g(k)$. The exact trajectory is determined by the noise and the disturbances, which are independent of the flooding, so given the exact loop $g(k)$ is a linear function of $\zeta(0), \dots, \zeta(k-1)$, hence Gaussian with mean zero and the covariance $\Sigma(k)$ of \cref{codesign:eq:cov}. We work on the intersection $E$ of eight events, one for each term of \cref{codesign:eq:eta}. Four are fixed for the mission: the key event, the coefficient event, the fiber event and the command-box event of \cref{codesign:ass:errors}, with complements of probability at most $\etakey$, $\eta_{\mathrm{coef}}$, $\eta_v$ and $\eta_\zeta$. The fifth, $E_{\mathrm{cyc}}$, is that every cycle's arithmetic stays within $\abar$, of probability at least $1 - K\etacyc$ by the conditional bound of \cref{codesign:ass:errors} and the tower property, since $z_{\mathrm{enc}}(k)$ is determined by $\mathcal{F}_k$. The sixth, $E_{\mathrm{ref}}$, is that every refresh stays within $\beta$ per entry, of probability at least $1 - R\,\etaref$ in the same way. The seventh, $E_{\mathrm{g}}$, is that $\norm{g(k)}_P \le \ghat$ for all $k$. The eighth, $E_{\mathrm{pos}}$, is that $\norm{g_i(k)} \le q\sqrt{\lambda_{\max}(\Sigma_i(k))}$ for all $i$ and $k$, where $g_i(k)$ is the position part of $g(k)$ on agent $i$. The last two have probabilities at least $1 - \etag$ and $1 - \etapos$ by the tail bounds at the end of this proof, so $E$ has probability at least $1 - \eta$ by the union bound. The flooding's entries inside $\abar_\nu$ are bounded on the command-box event, so no further term enters.

On $E$, while $z_{\mathrm{enc}}$ stays in the region, \cref{codesign:lem:additive} gives $r(k+1) = J_k\,r(k) + a(k) + \rho(k)$ with $\norm{a(k) + \rho(k)}_P \le \abar + \abar_\nu(\norm{d(k)}_P) + L_2 \norm{d(k)}_P^2 + \beta'\,\mathbf{1}[\text{a refresh follows cycle } k]$. The exact trajectory lies at $\norm{\cdot}_P$-distance at least $\rho_{\errbox}$ inside the convex region, so the Jacobians along it are those of maps that are Lipschitz on a neighbourhood: $\norm{J_{k+\ell-1} \cdots J_k}_P \le \gamma$ by the $\ell$-cycle contraction and $\norm{J_k}_P \le c$ by the one-cycle Lipschitz constant. Write $a^\ast = \abar + \abar_\nu(\dhat) + L_2 \dhat^2$. We show by induction over the blocks $\{j\ell, \dots, j\ell + \ell - 1\}$ that $\norm{r(j\ell)}_P \le r_{\mathrm{b}}(a^\ast, \beta', T)$ at every block start, that every $z_{\mathrm{enc}}(k)$ lies in the region, and that $\norm{r(k)}_P \le \rstar(a^\ast, \beta', T)$ and $\norm{d(k)}_P \le \dhat$ at every cycle. At $k = 0$ all hold because $d(0) = 0$. Suppose they hold up to the start $k = j\ell$ of a block. Inside the block, as long as the states so far lie in the region, the inputs of the cycles so far satisfy $\norm{a + \rho}_P \le a^\ast + \chi\beta'$, so unrolling the recursion over the first $i$ cycles of the block gives
\begin{align*}
\norm{r(j\ell + i)}_P &\le c^{\,i}\,\norm{r(j\ell)}_P + c_i\,a^\ast + c^{\,i-1}\,\chi_i\,\beta' \\
&\le \rstar(a^\ast, \beta', T), \qquad 0 \le i \le \ell - 1,
\end{align*}
Here $\chi_i$ is one if a refresh falls among the block's first $i$ cycles and zero otherwise, and at most one refresh falls in a block because $T \ge \ell$. The last inequality uses $c \ge 1$, $c_i \le c_{\ell-1}$ and $c^{\,i-1} \le c_{\ell-1}$. For $\ell = 1$ there is no cycle inside the block and the bound is $r_{\mathrm{b}}$. On $E_{\mathrm{g}}$ then $\norm{d(j\ell+i)}_P \le \ghat + \rstar(a^\ast, \beta', T) = \dhat \le \rho_{\errbox}$ by the choice of $\dhat$, so $z_{\mathrm{enc}}(j\ell + i)$ lies within $\rho_{\errbox}$ of the exact trajectory and in the region, which closes the induction inside the block. At the block's end,
\begin{align*}
\norm{r(j\ell + \ell)}_P &\le \norm{J_{j\ell+\ell-1} \cdots J_{j\ell}}_P\,\norm{r(j\ell)}_P \\
&\quad + \sum_{\iota=0}^{\ell-1} c^{\,\ell-1-\iota}\,\norm{a(j\ell+\iota) + \rho(j\ell+\iota)}_P \\
&\le \gamma\,\norm{r(j\ell)}_P + c_\ell\,a^\ast + c_\ell\,\chi_\ell\,\beta' .
\end{align*}
Unrolling over the blocks from $r(0) = 0$,
\[
\norm{r(j\ell)}_P \le c_\ell\,a^\ast \sum_{l \ge 0} \gamma^{\,l} + c_\ell\,\beta' \sum_{q \ge 0} \gamma^{\,q \lfloor T/\ell \rfloor} = r_{\mathrm{b}}(a^\ast, \beta', T) ,
\]
where the second sum runs over the blocks that contain a refresh, counted backwards from block $j$: two refreshes are $T$ cycles apart, so their blocks are at least $\lfloor T/\ell \rfloor$ apart. This closes the induction over the blocks.

For the position, $x_{\mathrm{enc},i}(k) - x_i(k)$ is the plant part of $d(k)$ on agent $i$, the sum of $g_i(k)$ and the plant part of $r(k)$ on agent $i$. The latter has Euclidean norm at most $\kappa_P \norm{r(k)}_P \le \kappa_P\,\rstar(a^\ast, \beta', T)$, and on $E_{\mathrm{pos}}$ the former has norm at most $q\sqrt{\lambda_{\max}(\Sigma_i(k))} \le \rg$. Their sum is \cref{codesign:eq:rpos}.

The two tail bounds remain. For $E_{\mathrm{pos}}$, $g_i(k) \sim \mathcal{N}(0, \Sigma_i(k))$ in $\R^3$, so $\norm{g_i(k)}^2 \le \lambda_{\max}(\Sigma_i(k))\,\xi$ with $\xi$ chi-square with three degrees of freedom. Its tail $\Pr[\xi > x] = 2\bar{\Phi}(\sqrt{x}) + \sqrt{2x/\pi}\,e^{-x/2}$ is at most $(\sqrt{2/(\pi x)} + \sqrt{2x/\pi})\,e^{-x/2}$ by the Mills ratio. We take $q^2$ as the $x$ at which this bound equals $\etapos/(MK)$, and the union over the $MK$ agent-cycles gives $E_{\mathrm{pos}}$. For $E_{\mathrm{g}}$, $\norm{g(k)}_P^2 = \sum_j a_j \xi_j^2$ with $a_j$ the eigenvalues of $P^{1/2}\Sigma(k)P^{1/2}$ and $\xi_j$ independent standard Gaussians. The moment generating function of $\sum_j a_j(\xi_j^2 - 1)$ is at most $\exp\bigl(\lambda^2\norm{a}_2^2/(1 - 2\lambda\norm{a}_\infty)\bigr)$ for $0 < \lambda < 1/(2\norm{a}_\infty)$, because $-\tfrac12\ln(1 - 2u) - u \le u^2/(1 - 2u)$ on $0 \le u < \tfrac12$, and the Chernoff bound at the best $\lambda$ gives $\Pr[\sum_j a_j(\xi_j^2 - 1) \ge 2\norm{a}_2\sqrt{t} + 2\norm{a}_\infty\,t] \le e^{-t}$. With $t = \ln(K/\etag)$ and the union over the cycles, $\ghat^2 = \max_k\,(\sum_j a_j + 2\norm{a}_2\sqrt{t} + 2\norm{a}_\infty\,t)$ gives $E_{\mathrm{g}}$.
\end{proof}

\begin{corollary}[Per-cycle envelope]
\label{codesign:cor:envelope}
On the event of \cref{codesign:thm:tube}, at every cycle $k = j\ell + p$ with $0 \le p \le \ell - 1$, agent $i$'s position deviates from the exact loop's by at most
\begin{align*}
&q(\etapos)\sqrt{\lambda_{\max}\bigl(\Sigma_i(k)\bigr)} + \kappa_P\bigl(c^{\,p}\,r_j + c_p\,a^\ast + c^{\,p-1}\chi_p\,\beta'\bigr), \\
&r_j = c_\ell\,a^\ast \sum_{l=0}^{j-1}\gamma^{\,l} + c_\ell\,\beta' \sum_{v=0}^{R_j - 1}\gamma^{\,v\lfloor T/\ell \rfloor},
\end{align*}
with $a^\ast$ and $\chi_p$ as in the proof and $R_j$ the number of refreshes before block $j$. The envelope never exceeds $\rpos$.
\end{corollary}

\begin{proof}
The proof of \cref{codesign:thm:tube} bounds $\norm{r(j\ell)}_P$ by the truncated sums $r_j$ before bounding them by $r_{\mathrm{b}}$, and $\norm{r(j\ell + p)}_P$ by $c^{\,p} r_j + c_p a^\ast + c^{\,p-1}\chi_p\beta'$. Its position step applies at each cycle, with $q\sqrt{\lambda_{\max}(\Sigma_i(k))}$ in place of its maximum $\rg$.
\end{proof}

\begin{proof}[Proof of \cref{codesign:cor:spec}]
On the event of \cref{codesign:thm:tube}, for every agent $i$ and $k \in \mathcal{K}$, $\norm{x_{\mathrm{enc},i}(k) - r_i(k)} \le \norm{x_i(k) - r_i(k)} + \norm{x_{\mathrm{enc},i}(k) - x_i(k)} \le \tol - \rpos + \rpos = \tol$.
\end{proof}

\begin{lemma}[The refresh through the prediction]
\label{codesign:lem:refreshpred}
A refresh after the release of cycle $k$ that changes the retained estimate by $\epsilon$, $\abs{\epsilon_a} \le \beta\,S_a$ with $S_a$ the encoding scales, leaves $e_x(k+1)$ and the command of cycle $k$ unchanged and moves $z(k+1)$ by $\bigl(0,\ (I - K_e C)\,\mathcal{H}(\xhat(k))\,\epsilon\bigr)$ to first order, $\mathcal{H}$ the state Jacobian of the coordinator's Heun step. Its $P$-norm is at most
\begin{align*}
\beta' = \sqrt{M}\Bigl[&\max_{\pm}\norm{\bigl(0,\ (I - K_e C)H_{x0}\,(\beta S \circ \pm 1)\bigr)}_{P_a} \\
&+ \Bigl(\max_{k,i}\norm{(I - K_e C)\,dH(r_i(k))}_{2 \to P_a} \\
&\qquad + \max_{\mathrm{vert}}\norm{(I - K_e C)\,dH(\cdot)}_{2 \to P_a}\Bigr)\norm{\beta S}_2\Bigr],
\end{align*}
the second term the field's part of $\mathcal{H}$ at the references and over the box, and the second-order term is part of the remainder.
\end{lemma}
\begin{proof}[Proof of \cref{codesign:lem:refreshpred}]
With $u(k)$ fixed, $\xhat(k+1) = p + K_e C\bigl(x(k+1) - p\bigr)$ with $p = \mathrm{Heun}_h\bigl(\xhat(k), u(k) + u^{\mathrm{ff}}(k)\bigr)$, so $\partial \xhat(k+1)/\partial \xhat(k) = (I - K_e C)\,\mathcal{H}(\xhat(k))$ with $\mathcal{H} = H_{x0} + dH(r) + dH(e_{\hat{x}})$ on the region; $x(k+1)$ and the command do not depend on the refreshed estimate. The triangle inequality over the three parts of $\mathcal{H}$ and the sum over the agents give $\beta'$.
\end{proof}

\begin{proof}[Proof of \cref{codesign:prop:delay}]
The conditions of the proposition are, on the region of \cref{codesign:ass:contraction} for $(e_x, e_{\hat{x}})$ with the stale feedback within the feedback's bound and for every $\tau \in \Theta$: (i) $\norm{J^{\tau}_k(z)}_P \le \gamma_\Theta < 1$; (ii) the remainder bound with a constant $L_2$; (iii) the command expansion of \cref{codesign:ass:contraction} with $B^{\tau}_k(z) = [\Psi_u(z);\ (I - K_e C)H_{u0} + K_e C\,\Psi_u(z);\ I]$, $\Psi_u$ the flow's sensitivity to the new command, and constants $L_B$, $L_\nu$; (iv) the delayed exact loop at $\norm{\cdot}_P$-distance at least $\rho_{\errbox}$ inside the region.

The stale feedback is a state, so the delayed loop is a loop of the form of \cref{codesign:sec:guarantee} on $\R^{M(2n+m)}$ with the maps $G^{\tau_k}_k$. By (i) and the mean value theorem on the convex region, each is a $\gamma_\Theta$-contraction in $\norm{\cdot}_P$, which is \cref{codesign:ass:contraction} with $\ell = c = 1$; (ii), (iii) and (iv) are its remaining hypotheses. The perturbation of the stale state at cycle $k + 1$ is the released command's error and flooding of cycle $k$, which (iii) already carries through the identity block of $B^{\tau}_k$, so \cref{codesign:lem:additive} and the proof of \cref{codesign:thm:tube} apply unchanged with $J_k$ and $B_k$ of the cycle's delay.

It remains to compute $\gamma_\Theta$. Write $J^{\tau}_k(z) = J^{\tau}_k(z_{\mathrm{ref}}(k)) + [J^{\tau}_k(z) - J^{\tau}_k(z_{\mathrm{ref}}(k))]$ with $z_{\mathrm{ref}}(k)$ the references ($e_x = e_{\hat{x}} = 0$, $s = 0$). At the references the flow runs under the feedforward alone, so its state Jacobian $\Psi_x$ does not depend on $\tau$, and the sensitivities to the two commands are $\Psi_s(\tau) = \int_0^\tau \Psi(h, \sigma)B\,d\sigma$ and $\Psi_u(\tau) = \int_\tau^h \Psi(h, \sigma)B\,d\sigma$. Hence $\norm{J^{\tau}_k(z_{\mathrm{ref}}) - J^{\tau'}_k(z_{\mathrm{ref}})}_P \le \abs{\tau - \tau'}\,\sup_{\sigma \in \Theta}\norm{S R\,\Psi(h, \sigma)B}\,(c_s + c_{Du})$, with $R = [I;\ K_e C;\ 0]$, $c_s = \norm{[0\ 0\ I]S^{-1}}$ and $c_{Du}$ the law's coupling in the norm, and $\Psi(h, \sigma)$ bounded by the Clohessy and Wiltshire flow plus its Gr\"onwall departure. The dense norm of the first term at a grid of $\Theta$ and at one cycle every 520 s, plus the difference to the nearest sample cycle, plus the grid's increment, bounds it for every $\tau \in \Theta$ and every cycle. The bracket is bounded on the box as in \cref{codesign:ass:contraction}: the plant rows' variation by the integral form of the flow's second variation with the Clohessy and Wiltshire transitions, the estimate rows' by the Heun step's vertex term, and the cubic through $\Psi_u$; its plant term is monotone in $\tau$, largest at an end of $\Theta$. For Heun's step on the polynomial field the Jacobian is affine in $(\tau, \tau^2)$, and the norm's convexity bounds it over $\Theta$ by its values at the four vertices of a polygon containing the arc $\{(\tau, \tau^2)\}$; this gives the model's $\gamma_\Theta$ without a grid.
\end{proof}

\begin{proof}[Proof of \cref{codesign:cor:linear}]
For $f(x) = Ax$ and a constant gain the exact map is affine, $G_k(z, \nu) = A_{\mathrm{cl}}\,z + B\,\nu + b_k$, with the references, the feedforward and the disturbances in $b_k$ and $A_{\mathrm{cl}}$ the closed-loop matrix of the joint estimator and controller dynamics. Its Jacobian is $A_{\mathrm{cl}}$ everywhere and its linearization is exact, so $L_2 = 0$. In a norm $\norm{\cdot}_P$ in which $\gamma = \norm{A_{\mathrm{cl}}}_P < 1$, every one-cycle map is a $\gamma$-contraction on all of $\R^{2Mn}$, so \cref{codesign:ass:contraction} holds with $\ell = 1$ and any region, and with $c = 1$, since a map with Lipschitz constant $\gamma < 1$ has Lipschitz constant 1. Then $c_1 = 1$, $c_0 = 0$, \cref{codesign:eq:cov} is the recursion stated, and \cref{codesign:eq:rstar,codesign:eq:rpos} give the two expressions.
\end{proof}

\begin{proof}[Proof of \cref{codesign:thm:design}]
Let the algorithm return $(\Delta, L, N, T, h)$. Its $T$ satisfies $T_{\min} \le T \le T_{\max}$. From $T \le T_{\max}$, \cref{codesign:eq:fit} holds, so the retained ciphertext keeps enough levels for $T$ cycles and the command branch above them, and the pipeline runs as \cref{codesign:sec:pipeline} states it with a refresh every $T$ cycles. From $T \ge T_{\min}$, $\dhat \le \rho_{\errbox}$ and $\rho_{\mathrm{ex}}(h) + \rpos \le \tol$ with the certificate's constants at the period $h$ and the error constants at $(\Delta, N)$. So \cref{codesign:thm:tube} applies, and \cref{codesign:cor:spec} gives the specification on an event of probability at least $1 - \eta$, with $\eta$ as computed. By construction $h \ge (1 + \mu)\max(t_1, t_T + \tref)$, so whenever every cycle's time and the refresh's stay within the factor $1 + \mu$ of their values at the design, every cycle, the one that carries the refresh included, completes within $h$. The menu is enumerated in nondecreasing order of each pair's own rule period, and the algorithm returns at the first pair whose interval $[T_{\min}, T_{\max}]$ is not empty. Every pair with a shorter period precedes it, and was found to violate one of \cref{codesign:eq:ring,codesign:eq:period,codesign:eq:tmax,codesign:eq:tmin} or to admit no certificate at its period. If the loop ends, every pair of the menu was examined and none satisfied them.
\end{proof}

\section{Depth and Counts per Stage}
\label{codesign:app:counts}

\Cref{codesign:tab:stages} lists the quantities of each stage as \cref{codesign:prop:cost,codesign:prop:counts} state them, with the trace's values for the worked example. The paragraphs below give the operations behind the counts.

\begin{proposition}[Per-cycle cost]
\label{codesign:prop:cost}
When $M$ is a power of two the pipeline spends $\dret = s\,\lceil \log_2(d_f + 1) \rceil + 1$ levels per cycle on the retained error and $\dcmd$ levels on the command branch as in \cref{codesign:eq:dcmd}, and it performs the products, rotations and additions per cycle of \cref{codesign:prop:counts}. It retains one ciphertext and refreshes it once every $T$ cycles, with $T$ the largest integer satisfying \cref{codesign:eq:fit}. For the worked example on the ring of sixteen agents this is $\dret = 5$, $\dcmd = 8$, 52 ciphertext products, 16 plaintext products and 38 rotations per cycle, and at $L = 39$ with $\Lref = 21$ for the two-iteration refresh it is $T = 3$, with equality in \cref{codesign:eq:fit}. When $M$ is not a power of two each layer of shifts adds one level and the per-offset rotations double.
\end{proposition}

\begin{proposition}[Operation counts]
\label{codesign:prop:counts}
When $M$ is a power of two the pipeline performs per cycle $s\,n_f + 1 + n_{K_e} + n_K + n_\varphi\,\abs{\offsets}$ ciphertext products, $M + 2\,m_{\offsets}$ plaintext products and $s\,\rho_f + \rho_{K_e} + \rho_K + 3\abs{\offsets} - 2\,h_{\offsets} + 1$ rotations, and the packing costs $4(M - 1)$ additions. Here $n_f$ and $\rho_f$ are the products and rotations of one evaluation of $f$, and $n_\varphi$ the products of one evaluation of $\varphi_3$. The counts $n_{K_e}$ and $n_K$ are the nonzero cyclic diagonals of the update's two blocks and of the law's two, and $\rho_{K_e}$ and $\rho_K$ the rotations among them. Finally $m_{\offsets}$ is the number of offsets whose cyclic diagonal is not constant, and $h_{\offsets} \in \{0, 1\}$ is one when $M/2 \in \offsets$.
\end{proposition}

\begin{table*}[t]
\caption{The stage quantities of \cref{codesign:prop:cost,codesign:prop:counts} for $M$ a power of two, and their values on the ring of sixteen.}
\label{codesign:tab:stages}
\centering
\small
\begin{tabular}{lccccc}
\toprule
Stage & Depth & Ct. products & Pt. products & Rotations & Additions \\
\midrule
Uploads and packing ($4M$ encryptions) & 0 & 0 & 0 & 0 & $4(M - 1)$ = 60 \\
Prediction ($s$ stages) & $s\lceil \log_2(d_f + 1) \rceil$ = 4 & $s\,n_f + 1$ = 25 & 0 & $s\,\rho_f$ = 14 & 23 \\
Update & 1 & $n_{K_e}$ = 12 & 0 & $\rho_{K_e}$ = 10 & 13 \\
Control law & $\lceil \log_2(\dphi + 1) \rceil$ = 2 & $n_K + n_\varphi \abs{\offsets}$ = 15 & $2 m_{\offsets}$ = 0 & $\rho_K + 3\abs{\offsets} - 2 h_{\offsets} + 1$ = 14 & 15 \\
Release (one encryption) & 1 & 0 & $M$ = 16 & 0 & $M$ = 16 \\
Refresh (every $T$ cycles) & $\Lref$ = 20 or 21 & -- & -- & -- & -- \\
\bottomrule
\end{tabular}
\end{table*}

\paragraph{Uploads} Under the realization that evaluates the field about the reference, the applied command is replaced by the applied feedback and a fifth vector, the reference's one-period defect, is added. The quantities uploaded once before the first cycle are the coefficients of $f$ with the integrator's step folded in, the estimator's gain, the gain $K$ and the coefficients of $\varphi$, each block-diagonal matrix as the ciphertexts of its nonzero cyclic diagonals. The shape of the computation that the coordinator holds in the clear (\cref{codesign:subsec:architecture}) is $M$, the offsets $\offsets$ with the $0/1$ masks of the graph's edges, the stages, and the indicator vectors of the release.

\paragraph{Packing} The fleet's $B = \nb M$ entries are a power of two when $M$ is, and the ciphertext then holds exactly them. The library encodes $N/(2B)$ copies of the block, so that its native rotation is the cyclic shift modulo $B$~\cite{albadawi2022openfhe}, and the rotation by $s\nb$ is the cyclic shift of the fleet's block. When $M$ is not a power of two, the fleet's block is not one either, the ciphertext's $B$ entries exceed it, and a rotation must not wrap the unused entries into the fleet's block. The cyclic-shift construction of Schlottke-Lakemper and Kholod~\cite{schlottkelakemper2024securecompute} then realizes each rotation by two rotations, two plaintext multiplications by indicator vectors and one addition, at one more level per layer of shifts~\cite{damera2026end}.

\paragraph{Prediction} One evaluation of $f$ forms the components of its argument by rotations within the block, one for every distinct offset between the position of a factor and that of the output component. A linear term of $f$ is one product of a coefficient ciphertext with a rotated copy of the state. A monomial of degree two is the product of one rotated copy with the sum of the coefficient products of its partner factors, the monomials sharing a factor grouped so that the outer product is formed once per group. For the worked example one evaluation is $n_f = 12$ products over $\rho_f = 7$ rotations, and the command term of the step one more product, so the prediction is 25 products and 14 rotations at four levels. At the first cycle the estimate is the first measurement with the reference's velocity.

\paragraph{Update} Each of the two blocks of \cref{codesign:eq:update} has six nonzero cyclic diagonals for the worked example, twelve products and ten rotations in all. A gain scheduled along the references costs the same per cycle, the diagonals of the cycle's gain uploaded in place of the constant ones.

\paragraph{Law} The odd part is computed once per offset. The coordinator forms the difference $\rot{s\nb}(\enc{e}) \ominus \enc{e}$, evaluates $\varphi_3$ on it with the coefficient ciphertext inside the product tree, and obtains the term of the opposite direction by oddness, $\varphi_3(\delta_{i,-s}) = -\varphi_3(\delta_{i-s,s})$, as one rotation of the result back by a block. For the cubic of the worked example this is three products and three rotations per offset, and an offset $s = M/2$, whose two directions coincide, costs one rotation. The command ciphertext is the sum of the linear and the odd part after one rotation that aligns the odd part's components with the command entries, the one rotation added in the count of \cref{codesign:prop:cost}.

\paragraph{Release} The library adds the flooding noise in the ciphertext's own units at decryption, so it has the mandated size only on a ciphertext at the first noise degree~\cite{albadawi2022openfhe}, which the addition of the encryption of zero produces at the same level. Under the library's rescaling a product leaves the ciphertext at the second noise degree, its scale squared, and the library applies the pending rescale at the next addition. The encryption of zero is formed at the level the product has after its rescale, so the sum is at the first degree and the flooding noise added at decryption has the mandated size.

\begin{remark}[Graphs that are not circulant]
\label{codesign:rem:general_graphs}
Any weighted graph is the sum of the cyclic diagonals of its adjacency matrix, one per nonzero offset, and it is circulant when every cyclic diagonal is constant. The coupling of a graph that is not circulant uses the same rotations, with the $0/1$ mask of the edges at each offset multiplying the rotated copy and the degree of each agent entering the diagonals of its tracking block, at no further level. The count of rotations is the number of nonzero cyclic diagonals, not the number of agents~\cite{damera2026end}. In row-major order the torus's vertical offset $\sqrt{M}$ has a constant cyclic diagonal, while its horizontal offsets $1$ and $\sqrt{M} - 1$ have cyclic diagonals with zeros where a row ends, so its Laplacian has seven nonzero cyclic diagonals against the ring's three~\cite{damera2026end}. The hypercube's offsets $2^b$ below $M/2$ have cyclic diagonals that select the agents whose bit $b$ is zero at $+2^b$ and one at $-2^b$, its offset $M/2$ is constant, and its Laplacian has $2 \log_2 M$ nonzero cyclic diagonals.
\end{remark}

\begin{remark}[Copies under sparse packing]
\label{codesign:rem:copies}
When the fleet occupies a fraction of the ciphertext's entries, the library's sparse encoding stores $G = N/(2B)$ copies of the fleet vector, and each decoded entry is the mean of its copies (\cref{codesign:lem:decoded}), so the error constants and the failure events of \cref{codesign:sec:guarantee} are per decoded entry, with no factor for the copies. Between refreshes a polynomial acts on each copy, so the decoded estimate follows the mean of the copies' maps rather than the map of their mean. The difference is second order in the per-copy deviations, which are $\sqrt{G} = 22.6$ times the decoded ones: at the mission the retained error's $P$-norm, $3.5 \times 10^{-3}$ m decoded, is at most $0.079$ m per copy, and $L_2 \times 0.079^2 = 1.9 \times 10^{-7}$ m per cycle against $\abar = 3.8 \times 10^{-3}$ m, inside the check of the second-order product term at $c_\times = \sqrt{3}\,G$ (\cref{codesign:app:percycle}).
\end{remark}

The retained chain is the prediction and the update, and the control law and the release are the command branch. For the worked example $n_f = 12$, $\rho_f = 7$, $n_{K_e} = n_K = 12$ (six nonzero cyclic diagonals in each of the four blocks), $\rho_{K_e} = \rho_K = 10$ and $n_\varphi = 3$. Across the families at sixteen agents the trace records 52, 55, 61, 58, 61 and 73 ciphertext products per cycle, for the ring, the 2- and 4-nearest rings, the torus, the hypercube and the complete graph in that order. The rotations per cycle are 38, 41, 47, 44, 45 and 57. The plaintext products are 16, 16, 16, 20, 22 and 16, the torus and the hypercube adding two mask products per offset whose cyclic diagonal is not constant. Every count agrees with \cref{codesign:prop:counts}.

\section{Configuration Tables}
\label{codesign:app:config}

\begin{figure}[t]
\centering
\includegraphics[width=\columnwidth]{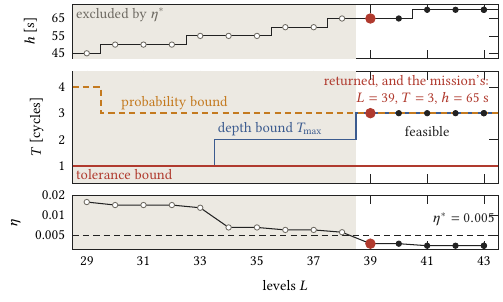}
\caption{The feasible set of \cref{codesign:alg:codesign} on the worked example against the levels $L$: the period (top), the three bounds on the refresh interval (middle), $\eta$ at $T_{\max}$ (bottom). Marked: the rule's output.}
\label{codesign:fig:feasible}
\end{figure}

Every encrypted run uses the configuration of \cref{codesign:tab:config}. The library prints it as HEStd\_128\_classic, ring $N = 131\,072$, batch 128, scaling FLEXIBLEAUTO, key UNIFORM\_TERNARY, scale bits 59 and first modulus bits 60, with large digits 3, an error standard deviation of 3.19 and $\log_2 Q = 2\,361$.

\begin{table}[H]
\caption{The configuration of every encrypted run, as the library prints it.}
\label{codesign:tab:config}
\centering
\small
\begin{tabular}{ll}
\toprule
Parameter & Value \\
\midrule
Setting, key & HEStd\_128\_classic, uniform ternary \\
Ring $N$, entries $B$ & $131\,072$, 128 \\
Scale $\log_2 \Delta$, first modulus & 59 bits, 60 bits \\
Levels $L$, depth of 1 iteration & 39, 20 at level budget $[3, 3]$ \\
Refresh & 2 iterations, precision 12, correction 8 \\
Depth of 2 iterations, $\Lref$ & 21 \\
\bottomrule
\end{tabular}
\end{table} \Cref{codesign:tab:encoding} gives the worked example's encodings, weights and constants. \Cref{codesign:tab:optimes} gives the measured times of the operations at that ring dimension and the refresh times. The trace's time model multiplies the operation times by the operation counts of \cref{codesign:app:counts}. The certificate's constants at each period of the menu are in \cref{codesign:tab:certificate}. \Cref{codesign:tab:delay} gives the delay grid. On the ring of sixteen at 65 s, in the norm of \cref{codesign:prop:delay} and on the plant's own sampled map, $\gamma = 0.9663$ undelayed, $0.9669$ at a fixed 22 s and $0.9668$ at a fixed 24 s, and the certified radius is 1.78 m undelayed, 1.83 m at a fixed 22 s and 1.88 m at a fixed 24 s, and the exact loop's held error plus the radius is 45.9, 45.7 and 45.7 m against 50 m.

\begin{table*}[t]
\caption{The delay grid: the certificate undelayed (the paper's norm) and at 22 s (the norm of \cref{codesign:prop:delay}), at each row's estimate; the ring of 128 undelayed in that norm.}
\label{codesign:tab:delay}
\centering
\footnotesize
\setlength{\tabcolsep}{4pt}
\begin{tabular}{lcccc}
\toprule
Graph, $M$, $h$, $w$ & $\gamma$ 0 / 22 s & $\rpos$ 0 / 22 s [m] & $\rho_{\mathrm{ex}} + \rpos$ 0 / 22 s [m] & 50 m \\
\midrule
Ring, 16, 65 s, 1 & 0.9834 / 0.9669 & 1.91 / 1.83 & 46.05 / 45.65 & pass \\
Ring, 32, 65 s, 1 & 0.9647 / 0.9443 & 0.41 / 0.33 & 10.39 / 10.75 & pass \\
Ring, 64, 65 s, 1 & 0.9667 / 0.9557 & 1.09 / 0.75 & 13.82 / 13.88 & pass \\
Ring, 128, 75 s, 1 & 0.9600 / 0.9693 & 9.99 / 13.69 & 29.47 / 33.20 & pass \\
2-nearest ring, 16, 65 s, 1 & 0.9879 / 0.9858 & 1.50 / 1.48 & 47.74 / 47.98 & pass \\
4-nearest ring, 16, 65 s, 0.5 & 0.9859 / 0.9832 & 1.53 / 1.64 & 49.87 / 50.27 & pass / fail by 0.27 \\
Torus, 16, 65 s, 1 & 0.9920 / 0.9770 & 2.48 / 1.52 & 50.75 / 50.06 & fail (tolerance) \\
Hypercube, 16, 65 s, 1 & 0.9920 / 0.9781 & 3.15 / 1.97 & 49.32 / 48.06 & pass \\
Complete, 16, 65 s, 0.5 & 0.9921 / 0.9781 & 5.25 / 2.84 & 59.55 / 57.32 & fail (held error) \\
\bottomrule
\end{tabular}
\end{table*}

\begin{table*}[t]
\caption{The worked example's encodings and constants at $h = 65$ s.}
\label{codesign:tab:encoding}
\centering
\small
\begin{tabular}{ll}
\toprule
Quantity & Value \\
\midrule
Encoded unit, retained $\hat{x}_i - r_i$, references, measurements & 1 km (positions), 1 m/s (velocities) \\
Encoded unit, commands $S_u$ & $0.02$ m/s$^2$ \\
Field coefficients uploaded & $\tfrac{h}{2}\, a\, \prod_k S_k / S_j$ for a term $a \prod_k x_k$ of component $j$ \\
Cubic coefficient uploaded & $a_3 S_2^3 / S_{u,2}$ times the edge weight \\
Entries & 8 per agent (6 used), $B = 8M = 128$ \\
\midrule
LQ weights & $Q = \operatorname{diag}(500^{-2}\,\mathrm{m}^{-2} I_3,\ 0.5^{-2}\,(\mathrm{m/s})^{-2} I_3)$, $R = (10^{-3})^{-2}\,(\mathrm{m/s^2})^{-2} I_3$ \\
Model of the design & the zero-order-hold sampling at $h$ of the Clohessy and Wiltshire part \\
Kalman covariances & $W = \operatorname{diag}(1\,\mathrm{m}^2 I_3,\ (6 \times 10^{-3}\,\mathrm{m/s})^2 I_3)$, $V = (3\,\mathrm{m})^2 I_3$ \\
Coupling & $K_c = K_g / 2$, $w_s = 1$ (ring; $0.5$ for the complete graph) \\
Cubic & $a_3 = K_{c,22} / (5\,\mathrm{km})^2 = 2.23 \times 10^{-14}\,\mathrm{s}^{-2}\mathrm{m}^{-2}$ \\
Noise and disturbance & measurement 3 m per axis clipped at 9 m; process at most $10^{-6}$ m/s$^2$ per axis \\
\midrule
Security, ring, scale & HEStd\_128\_classic, $N = 131\,072$, $\Delta = 2^{59}$, first modulus 60 bits \\
Depth, refresh & 39 = 18 levels + 21, budget $[3,3]$, 2 iterations at precision 12, correction factor 8, $T = 3$ \\
Key switching, error & 3 large digits, $\sigma = 3.19$, $\log_2 Q = 2\,361$, uniform ternary key, flexible rescaling \\
Flooding & $s = 30$, $n_{\mathrm{d}} = 1$ (given), $\hat{\varepsilon} = 2^{25.87}$ (estimation phase of 30 cycles) \\
\midrule
Units of the constants & $\beta$ per entry in encoded units; $\abar$, $\bar{\beta}$, $\dhat$, $\rho_{\errbox}$ in $\norm{\cdot}_P$; $\sigfl$ in m/s$^2$ per released entry; $\kappa_P$ in m \\
\bottomrule
\end{tabular}
\end{table*}

\begin{table}[H]
\caption{The certificate of the ring of sixteen per period of the menu.}
\label{codesign:tab:certificate}
\centering
\footnotesize
\setlength{\tabcolsep}{4pt}
\begin{tabular}{rrrrrrrr}
\toprule
$h$ [s] & $\gamma$ & $\kappa_P$ [m] & $\rho_{\errbox}$ & $\rho_{\mathrm{ex}}$ [m] & $\hat{\varepsilon}$ [bits] & $\rpos$ [m] & margin [m] \\
\midrule
20 & 0.9946 & 69.6 & 0.826 & 41.35 & 28.30 & 5.64 & 8.65 \\
40 & 0.9894 & 97.6 & 0.590 & 42.87 & 26.53 & 2.35 & 7.13 \\
65 & 0.9834 & 123.1 & 0.468 & 44.14 & 25.90 & 1.91 & 5.86 \\
90 & 0.9782 & 143.3 & 0.402 & 46.34 & 25.50 & 1.74 & 3.66 \\
130 & 0.9731 & 169.6 & 0.339 & 49.22 & 25.05 & 1.71 & 0.78 \\
\bottomrule
\end{tabular}
\end{table}

\begin{table}[H]
\caption{Measured operation times at $N = 131\,072$ on sixteen threads, medians of three.}
\label{codesign:tab:optimes}
\centering
\small
\begin{tabular}{lrrrrr}
\toprule
Operation & 0 & 20 & 25 & 30 & 38 \\
\midrule
Encode & 11.9 & 4.6 & 2.6 & 2.1 & 1.6 \\
Encrypt & 355 & 142 & 109 & 79 & 33 \\
Decrypt & 135 & 52 & 40 & 29 & 17 \\
Addition of ciphertexts & 79 & 9.1 & 6.2 & 3.7 & 0.5 \\
Product of ciphertexts & 466 & 147 & 113 & 63 & 40 \\
Product by a plaintext & 80 & 8.7 & 5.8 & 3.6 & 0.8 \\
Rotation & 518 & 138 & 119 & 61 & 38 \\
Rescale & 34 & 11 & 7.8 & 7.3 & 7.4 \\
\midrule
Refresh, one iteration & \multicolumn{5}{c}{21.5 s} \\
Refresh, two iterations & \multicolumn{5}{c}{43.8 s} \\
\bottomrule
\end{tabular}
\end{table}

\section{Experimental Details}
\label{codesign:app:experiments}

This appendix gives the setup and the pass criteria of each experiment of \cref{codesign:sec:evaluation}, fixed before its numbers, and the checks the section summarizes. The refresh error's samples and criteria are in \cref{codesign:app:refresh}.

\subsection{Trace validation}

\paragraph{Setup} The traced record and the instrumented library's record are compared operation by operation. The circuits are a grid of every operation at every level and noise degree, multi-ciphertext cyclic shifts, the coupling of the six graph families under full and sparse packing, and the full cycle of the pipeline with its refreshes.

\paragraph{Pass} Zero mismatches in kind, level, noise degree and rotation index, and decrypted outputs equal to the traced values to the encryption error.

\subsection{The certificate}

\paragraph{Setup} For the worked example and each family, the constants of \cref{codesign:ass:contraction} are computed at each period of the menu as in \cref{codesign:subsec:assumptions}, on its box $\errbox$. The Jacobian's variation over the box is enumerated at its vertices for the field's terms and bounded through the Laplacian for the cubic, and $L_2$ follows from their second derivatives. The exact loop is then simulated from the edge of the box, and the tube of \cref{codesign:thm:tube} is evaluated at the mission's configuration with the covariance of \cref{codesign:eq:cov} along the exact trajectory.

\paragraph{Pass} Five criteria. $\gamma \le 0.99$ at the design period on the ring of sixteen. The simulated decay per cycle from the box's edge is at most $\gamma$. The box contains the exact loop's and the mission's formation errors. $\rpos$ is within the exact loop's margin at the mission's configuration, with the measured deviation inside it. The exact loop driven by Gaussian command noise of the mandated size deviates within 25 percent of what \cref{codesign:eq:cov} predicts.

\paragraph{The sweep} On the ring of sixteen at a coupling weight of 0.1 the rule returns 45 s with a refresh every cycle above 48.2 m, then 55 s, then 65 s, then 70 s with 41 levels and a refresh every third cycle down to 47.5 m, and nothing below. At $N = 262\,144$ the refresh's failure probability summed over the mission exceeds one, so that ring never returns under a target. Without the probability target the rule returns 29 levels, a refresh every cycle and 45 s, feasible at 50 m with a radius of 2.60 m against the mission's 1.91 m.

\paragraph{The families} At 65 s the 2-nearest ring and the hypercube are certified at their instance weight. The torus misses at weight 1 through the tube itself and is certified at 0.25. The 4-nearest ring exceeds the tolerance in exact arithmetic at weight 1 and is certified at 0.25 and 0.5 with the per-cycle model's estimate at those weights. The complete graph has no configuration at 50 m, through the tube at weight 0.25 and through its exact loop at weight 0.5.

\paragraph{Checks} At 65 s, $\gamma = 0.983$ is within 0.009 of the largest Jacobian norm sampled, and the decay from the box's edge is at most 0.961 per cycle. The contraction per unit of time is within 0.004 of 0.983 per 65 s at every period, because the design's weights set it and the period only changes how it is counted. The exact loop driven by Gaussian command noise of the mandated size deviates by 0.328 m root mean square, the value \cref{codesign:eq:cov} predicts. At the probe's period of 20 s the radius is 5.64 m against 8.65 m, with the measured 0.935 m inside. The ring of sixteen rerun for six cycles gives $2^{25.93}$ against $2^{25.9}$ over its thirty, so the six-cycle estimates compare like for like.

\subsection{Encrypted closed loops across graphs}

\paragraph{Setup} \Cref{codesign:alg:codesign} is run for each graph family and fleet size. At 32, 64 and 128 agents the position weight of the tracking design is 250 m in place of 500 m. On the ring of sixteen the period of \cref{codesign:eq:period} is 62.5 s before the 5 s grid. The encrypted loop is run for the mission's length on the ring of sixteen agents, and for every family and size for a few cycles at the same setting, which measures the cycle and refresh times that the trace's time model predicts.

\paragraph{Pass} Four criteria, all required. The cycle time, refresh cycles included, is within $h$ on every cycle whose measured times stay within the margin $\mu$ of their design values (\cref{codesign:thm:design}). The largest deviation from the plain loop with the same disturbances is at most $\rpos$. The formation error is within $\tol$ after $k_0$. The refreshes per cycle equal $1/T$.

\paragraph{The mission's times and the families' runs} In the mission on the ring of sixteen the commands deviated by at most $2.1 \times 10^{-5}$ m/s$^2$ and the velocities by at most $1.7 \times 10^{-3}$ m/s. The coordinator's cycle took 15.0 s at the median, 21.4 s in the cycle after a refresh, when the retained ciphertext is at its highest level, and about 14 s in the cycle before one. The 173 refresh cycles, one in three as designed, took 58.4 s of the 65 s period at the median. The one cycle past the period is the refresh cycle whose refresh was the mission's longest, 50.6 s against the 44.0 s median, which ran 1.7 s past the period. The agents' encryptions, which the run performs serially at 20.7 s per cycle, are 1.3 s per agent in parallel and outside the coordinator's period. The families' six cycles are the first interval of a fresh retained ciphertext, whose sixth cycle runs at the levels of the cycle after a refresh. The ring's sampled mission cycle at the same position took 21.7 s, against the median of 21.4 s. The prediction and the update, which do not depend on the graph, take about four fifths of a cycle, while the law's rotations grow from 38 to 57 per cycle from the ring to the complete graph. The refresh is the ring's, the context being the same. The deviations from the exact loop after six cycles were 0.12 to 0.30 m against the ring's 0.29 m at that cycle. With the released-error estimates of \cref{codesign:tab:familiesfull}, the per-family flooding differs from the ring's by about a factor of two in standard deviation. The per-cycle model's estimates order the families as the measured ones do. At 32 and 64 agents the encrypted runs' per-cycle deviations lie inside the emulator's 5 to 95 percent band at 9 of 12 and 12 of 12 cycles. Per refresh the failure probability is $4.2 \times 10^{-5}$, $8.4 \times 10^{-5}$ and $1.7 \times 10^{-4}$ at 32, 64 and 128 agents. In the emulated missions, which delay each command by the ring mission's measured cycle times scaled by the family's ratio from the timing runs, the delay raises the deviation's median by 3 percent on the ring and up to 21 percent on the complete graph, in the order of the families' cycle times.

\begin{table}[t]
\caption{The six graph families at sixteen agents and 65 s.}
\label{codesign:tab:families}
\centering
\footnotesize
\setlength{\tabcolsep}{2.4pt}
\begin{tabular}{lccccccc}
\toprule
Family & $\abs{\offsets}$ & $w$ & $\gamma$ & $\rho_{\mathrm{ex}}$ [m] & $\rpos$ [m] & margin [m] & 50 m \\
\midrule
Ring & 1 & 0.25 & 0.978 & 45.77 & 1.42 & 4.23 & pass \\
 & & 1.0 & 0.983 & 44.14 & 1.91 & 5.86 & pass \\
2-nearest ring & 2 & 0.25 & 0.979 & 46.21 & 1.02 & 3.79 & pass \\
 & & 1.0 & 0.988 & 46.25 & 1.50 & 3.75 & pass \\
4-nearest ring & 4 & 0.25 & 0.981 & 48.42 & 1.39 & 1.58 & pass \\
 & & 0.5 & 0.986 & 48.34 & 1.53 & 1.66 & pass \\
Torus & 3 & 0.25 & 0.980 & 47.93 & 1.27 & 2.07 & pass \\
 & & 1.0 & 0.992 & 48.26 & 2.48 & 1.74 & fail \\
Hypercube & 4 & 0.25 & 0.980 & 47.22 & 1.62 & 2.78 & pass \\
 & & 1.0 & 0.992 & 46.17 & 3.15 & 3.83 & pass \\
Complete & 8 & 0.25 & 0.983 & 48.29 & 2.40 & 1.71 & fail \\
 & & 0.5 & 0.992 & 54.30 & 5.25 & $-4.30$ & fail \\
\bottomrule
\end{tabular}
\end{table}

\begin{table*}[t]
\caption{The six graph families at sixteen agents and 65 s beside \cref{codesign:tab:families}: the released-error estimate ($\dagger$ measured), the rule's output at 50 m, the cycle time, and the emulated mission with each command delayed by the ring mission's measured cycle times scaled by the family's ratio from the timing runs (the delayed exact loop's held error, the median of the largest position deviation and its 95th percentile).}
\label{codesign:tab:familiesfull}
\centering
\footnotesize
\setlength{\tabcolsep}{3.5pt}
\begin{tabular}{lcllll}
\toprule
Family & $w$ & $\hat{\varepsilon}$ meas., used [bits] & rule at 50 m, $\eta^\ast$ 0.005 & $\tcyc$ meas. / pred. [s] & emulated, delayed [m] \\
\midrule
Ring & 0.25 & 25.93, 25.16 & mission's, 65 s & 21.8 / 18.5 & \\
 & 1.0 & 25.90$^{\dagger}$ & mission's, 65 s & & 44.7, 0.85 (0.96) \\
2-nearest ring & 0.25 & 25.55, 24.78 & mission's, 65 s & 21.2 / 19.2 & \\
 & 1.0 & 25.55$^{\dagger}$ & mission's, 65 s & & 44.9, 0.52 (0.60) \\
4-nearest ring & 0.25 & 26.36, 25.56 & mission's, 65 s & 21.8 / 20.8 & \\
 & 0.5 & 26.36, 25.87 & mission's, 65 s & & 46.0, 0.61 (0.71) \\
Torus & 0.25 & 25.88, 25.15 & mission's, 65 s & 23.0 / 19.8 & \\
 & 1.0 & 25.88$^{\dagger}$ & none, tol. 50.7 m & & 45.0, 0.63 (0.73) \\
Hypercube & 0.25 & 26.28, 25.56 & mission's, 65 s & 22.2 / 20.3 & \\
 & 1.0 & 26.28$^{\dagger}$ & mission's, 65 s & & 45.2, 0.84 (0.99) \\
Complete & 0.25 & 26.98, 26.66 & none, tol. 50.7 m & 24.4 / 23.5 & 46.2, 1.07 (1.25) \\
 & 0.5 & 26.98$^{\dagger}$ & none, held 54.3 m & & 44.1, 1.09 (1.28) \\
\bottomrule
\end{tabular}
\end{table*}

\section{Glossary of Encryption Terms}
\label{codesign:app:terms}

\begin{figure*}[t]
\centering
\includegraphics[width=\textwidth]{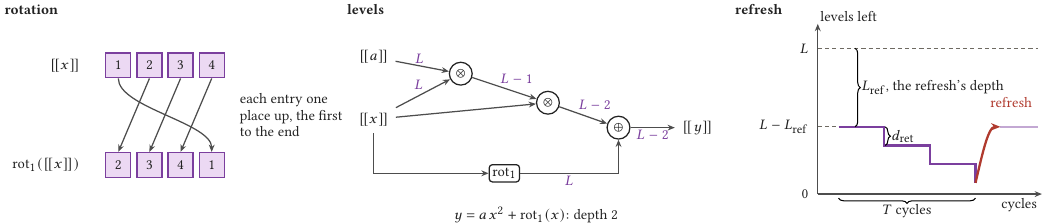}
\caption{A rotation by one entry (left), the levels of a small circuit (middle) and the refresh on the levels of the retained ciphertext (right).}
\label{codesign:fig:terms}
\end{figure*}

\Cref{codesign:fig:terms} draws three of the terms below.

\paragraph{The packed vector and its entries} One ciphertext of the scheme encrypts a vector of $N/2$ real numbers at once, and every operation acts on all entries together. The pipeline puts the whole fleet in one such vector, agent $i$ in a block of $\nb$ consecutive entries, so that one addition or one product acts on every agent at once. Entries the fleet does not use are zero.

\paragraph{Rotation} The rotation $\rot{s}$ moves every entry $s$ positions up and wraps the first $s$ entries to the end: $\rot{1}$ of the encryption of $(1, 2, 3, 4)$ is the encryption of $(2, 3, 4, 1)$. It is the only way to move data between entries, and it needs a key of its own for each $s$. A rotation by $s$ whole blocks of $\nb$ entries places agent $i + s$'s data under agent $i$'s for every $i$ at once. This is how the coupling of \cref{codesign:eq:law} is computed, and how a neighbour's state reaches an agent's entries without the coordinator knowing either.

\paragraph{Level and depth} Every product of two encrypted vectors, or of an encrypted vector and a plaintext one, consumes one of a fixed number of levels that the ciphertext starts with; additions and rotations consume none. The depth of a computation is the longest chain of products from an input to its output~\cite{schluter2024code}, and a ciphertext that has consumed its levels can be added and rotated but not multiplied. A polynomial of degree $d$ with encrypted coefficients, computed as a balanced tree of products, has depth $\lceil \log_2(d + 1) \rceil$, the coefficient counting as one more factor.

\paragraph{The refresh} The refresh, which the scheme's literature calls bootstrapping, restores levels to a ciphertext without decrypting it. It is itself a computation of fixed depth, so it returns fewer levels than a fresh ciphertext has, and it perturbs every entry by a small error with a small failure probability. At standard security it takes tens of seconds, the longest of any operation, and is the only one whose time is comparable to a control period.

\paragraph{Matrices and polynomials} A block-diagonal matrix times an encrypted vector is computed by the diagonal method, a sum over the matrix's nonzero cyclic diagonals of products with rotated copies~\cite{halevi2014algorithms}, whether the matrix is a plaintext or an upload of its cyclic diagonals. Comparisons, divisions and roots are not available, so a controller on ciphertexts is a polynomial of its inputs, which is why $\varphi$ is a polynomial and the estimator's gain is not computed by inversion. Each operation adds its own error to the ciphertext, and after a product of two ciphertexts the error grows with the size of the operands, which is why the per-cycle error of \cref{codesign:sec:guarantee} is a bound on the region $\region$.

\paragraph{The two-iteration refresh} The library offers a second iteration of the refresh, which refreshes the error of the first and lowers it by orders of magnitude at the cost of one more level and twice the time~\cite{albadawi2022openfhe}. \Cref{codesign:app:refresh} derives its $\beta$ and $\etaref$ from the library's circuit.

\paragraph{Scale and error} Real numbers enter the scheme multiplied by a scaling factor $\Delta$ and rounded to integers, so every value has a rounding error of about $1/\Delta$, and every operation adds a little more. The decryption of a value is therefore the value plus an encryption error, which the closed loop sees as a bounded disturbance. A larger $\Delta$ means a smaller error and a larger modulus.

\paragraph{Ring dimension and the standard's table} The ciphertext's size $N$ is the ring dimension. It sets the number of entries to $N/2$ and the time of every operation. For a given security level the standard's table gives the smallest $N$ whose modulus admits the levels and the scaling factor the computation needs, so more levels or a finer scale eventually cost a larger $N$ and a slower loop.

\end{document}